\documentclass{article}
\usepackage{amsmath,amssymb,amsthm}
\RequirePackage[colorlinks,citecolor=blue,urlcolor=blue]{hyperref}
\usepackage[sort&compress,square,comma]{natbib}
\usepackage[T1]{fontenc}
\usepackage[utf8]{inputenc}
\usepackage{bm} %boldmath
\usepackage{esint}
\usepackage{xspace}
\usepackage{bm}
  \usepackage[pdftex]{xcolor}
  \usepackage[pdftex]{graphicx}
  \DeclareGraphicsExtensions{.jpg,.gif,.pdf,.png}
\numberwithin{figure}{section}
\numberwithin{equation}{section}
{\left\{ \begin{array}{l@{\quad}l}}%
{\end{array} \right.}

\global\long\def\E{\text{\textbf{E}\,}}
\global\long\def\P{\text{\textbf{P}}}
\global\long\def\CandSet{\mathcal{L}}

\renewcommand{\rho}{\varrho}

\newcommand{\R}{\mathbb{R}}
\newcommand{\N}{\mathbb{N}}
\newcommand{\NN}{\mathcal{N}}
\renewcommand{\L}{\mathcal{L}}

\newcommand{\X}{\bm{X}}
\newcommand{\V}{\bm{V}}

\newcommand{\MU}{\bm{\mu}}
\newcommand{\NU}{\bm{\nu}}
\newcommand{\SIGMA}{\bm{\Sigma}}
\newcommand{\LAMBDA}{\bm{\Lambda}}

\newcommand{\x}{\bm{x}}
\newcommand{\xx}{\overline{\bm{x}}}
\newcommand{\y}{\bm{y}}

\newcommand{\n}{\bm{n}}

\renewcommand{\t}{\bm{t}}
\newcommand{\q}{\bm{q}}
\newcommand{\Pkt}{\,\mbox{\tiny$\circ$}\,}

\newcommand{\cov}{\text{cov}}

\newcommand{\iid}{\text{i.i.d. }}

\newcommand{\DPW}{\textsc{Dpw}\xspace}
\newcommand{\GreedyOCBA}{\textsc{GreedyOCBA}\xspace}
\newcommand{\EqAlloc}{\textsc{EqAlloc}\xspace}
\newcommand{\PLUCK}{\textsc{Pluck}\xspace}
\newcommand{\BayesRS}{\textsc{BayesRS}\xspace}
\newcommand{\KN}{$ {\cal{KN\negthickspace+\negthickspace+}} $ }

\theoremstyle{definition}
\newtheorem{Theorem}{Theorem}
\newtheorem{Example}{Example}

\begin{document}

% "Title of the paper"
\title{Ranking and Selection: A New Sequential Bayesian
    Procedure for Use with Common Random Numbers} 
\author{Bj{\"o}rn G{\"o}rder\thanks{Institute of Applied Stochastics an OR, TU
    Clausthal, Germany, bgoerder@gmail.com}
\and
Michael Kolonko \thanks{Institute of Applied Stochastics an OR, TU
    Clausthal, Germany, kolonko@math.tu-clausthal.de}}
\maketitle

\begin{abstract}
  We introduce a new concept for selecting the best alternative out of a given
  set of systems which are evaluated with respect to their expected
  performances. We assume that the systems are simulated on a computer and that
  a joint observation of all systems has a multivariate normal distribution with
  \emph{unknown mean} and \emph{unknown covariance} matrix. In particular, the
  observations of the systems may be stochastically dependent as it is the case
  if common random numbers are used for the simulation. The main application we
  have in mind is heuristic stochastic optimization where ‘systems’ are
  different solutions to an optimization problem with random inputs.

  We use a Bayesian setup with an uninformative prior and iteratively allocate a
  fixed number of simulations based on the posterior distribution of the
  observations until the ranking and selection decision is correct with a given
  high probability.  We introduce a new simple allocation strategy that is
  directly connected to the error probabilities calculated before. The necessary
  posterior distributions can only be approximated, but we give extensive
  empirical evidence that the error made is well below the given bounds.

  Our extensive test results show that our procedure \BayesRS uses less simulations than
  comparable procedures from the literature in different correlation scenarios.  At
  the same time \BayesRS needs no additional prior parameters and can cope with
  different types of ranking and selection tasks.

 \textbf{keywords:}  Sequential Ranking and Selection,
% Common Random Numbers, Bayesian Statistics, Multiple Testing, Missing Data
\end{abstract}

% \begin{keyword}[class=MSC]
% \kwd[Primary ]{62F07 } %Ranking and Selection
% \kwd{62F15} % Bayesian Inference
% \kwd[; secondary ]{11K45}%   	Pseudo-random numbers; Monte Carlo methods
% %\kwd{}
% %\kwd{}
% \end{keyword}

% \end{frontmatter}

% AOS,AOAS: If there are supplements please fill:
%\begin{supplement}[id=suppA]
%  \sname{Supplement A}
%  \stitle{Title}
%  \slink[doi]{10.1214/00-AOASXXXXSUPP}
%  \sdatatype{.pdf}" 
%  \sdescription{Some text}
%\end{supplement}

%\end{document}

%%%%%%%%%%%%%%%%%%%%%%%%%%%%%%%%%%%%%%%%%%%%%%%%%%

\section{Introduction}

We consider ranking and selection of systems based on the average performance of
the alternatives. The particular set-up used here is motivated by problems from
optimization under uncertainty.

Often in operations research as well as in technical applications the
performance of solutions depends on some random influence like market
conditions, material quality or simply measurement errors. We shall call such
random influences a random \emph{scenario}. Usually, the aim is then to find a
solution with minimal expected costs taken over all scenarios.  Let $
X_{i}=c(i,Z)$ be the random costs when solution (or 'alternative') $ i$ is
applied to the random scenario $ Z$. We want to measure the quality of $ i$ by
its expected costs $\mu_i:=\E c(i,Z)$ taken over all possible scenarios.

Except for particularly simple cases, we will not be able to calculate this
expression analytically.  Instead, we have to estimate $ \E c(i,Z)$ based on a
sample $ c(i,z_1) ,\ldots, c(i,z_n),$ where $ z_1 ,\ldots, z_n$ are random
scenarios. We assume here that simulations are done on a computer, therefore we
can identify $ z_1 ,\ldots, z_n$ with the seeds used for the random generator.

Optimization with respect to a simulated cost function is often done by
heuristic search methods like genetic algorithms, ant algorithms or
cross-entropy optimization (see e.g.\ \cite{reeves2010genetic},
\cite{dorigo2010ant}, \cite{wu2014asymptotic}). Typically, these methods take a
relatively small set $ {\cal L}:=\{1 ,\ldots, L\}$ of solutions (a
`population') and try to improve the quality of $ \cal L$ iteratively.  The
improvement step usually includes a \emph{selection} of the best solutions from
$ \cal L$ with respect to their expected costs $ \mu_i:=\E c(i,Z)$. Methods of
ranking and selection are therefore widely used in heuristic optimization under
uncertainty, see e.g. \cite{Schmidt2006} for an overview. As the $ \mu_i$ can
only be estimated, selection will return sub-optimal solutions with a certain
error probability.

The aim of this paper is to develop a strategy that allocates simulation runs to
alternatives in such a way that the error probability for the selection of good
alternatives is below a given bound. 

 It is well known, that if  observations of different solutions are
\emph{positively correlated}, it is more efficient to use common random numbers
(CRN), i.e.\ to compare the solutions on the same scenario (see
e.g. \cite{glasserman1992some}). Positive correlation in our case roughly means,
that if a solution $ i\in {\cal L} $ has, for a scenario $z,$ costs $ c(i,z)$
that are above average, then costs $ c(j,z)$ will tend to be over average for
all the other solutions $ j\in \cal L$ also. In other words, if some scenario $
z$ is relatively difficult (costly) for some solution $ i,$ then $ z$ will tend to be
difficult for all solutions.  Similarly, a scenario that has small costs (below
average) for one solution will tend to be an easy scenario for all solutions
creating smaller costs for all of them. This is the behavior which we would
expect for many optimization problems. Therefore, we are interested in ranking
and selection strategies that can deal with dependent observation from
simulation with CRN.

There are many different approaches to ranking and selection in the literature,
we mention only the most prominent ones here.  Many papers concentrate on
sampling from \emph{independent} alternatives, see e.g.\ \citep{KimNelsonUeber}
for an overview and \citep{kim2006asymptotic} for an advanced sequential method,
\KN. Other authors allow \emph{correlated samples} but assume that the
correlation structure, i.e.\ the covariance matrix of the joint distribution is
known as e.g.\ in \citep{Fu2007a}, or, in a Bayesian set-up that (part of) the
prior distribution is known (\citep{frazier2011value},\citep{qu2012ranking}).

As this may not be the case in realistic applications, many procedures use a
\emph{two stage} approach in which a sample of fixed size $ n_0$ is taken from
all observations in the first stage. The unknown parameters, as e.g.\ the
covariance matrix, are then estimated from the first stage sample. Based on these
estimates, the samples in the next stage are allocated to the alternatives. In a
pure two-stage procedure, sampling stops after the second stage and ranking and
selection is performed as e.g.\ in \cite{nelson1995using}, \cite{Chick2001},
\cite{Fu2007a} or \cite{peng2012}. A \emph{sequential} procedure updates the
estimates after each stage and starts a new iteration until the final decision
reaches a certain level of quality, see e.g.\ \cite{kim2006asymptotic},
\cite{Chick2001a} or \cite{qu2012ranking}.

  The quality of the final decision may be measured by different functions: the
\emph{opportunity cost} (or linear loss) is the difference between the mean of
the selected alternative and the true best value (e.g.\ \cite{Chick2001a},
\cite{qu2012ranking}) whereas the \emph{$ 0-1$-reward} function is $ 1$ if the
true value is selected and $ 0$ if not. Here, a selection may be seen as correct
if it is within a $ \delta$-distance of the true correct alternative
(indifference zone). The expected value of the $ 0-1$-reward function is the
probability of a correct selection (PCS), this is used e.g.\ in
\cite{Chick2001a}, \cite{Fu2007a} or \cite{peng2012}.

Similarly, the sample allocation on a single stage may be determined using the
\emph{value-of-information} (or \emph{knowledge gradient}) approach, choosing an
alternative that promises the largest expected increase in the best mean value
after the next observation, this is used e.g.\ in \cite{frazier2011value} or
\cite{qu2012ranking}. On the other hand, the goal may be to maximize the
\emph{probability of a correct selection} (PCS) given a fixed budget of
simulations (\cite{Chen1997},\cite{Chen2010}, \cite{peng2012} or
\cite{Fu2007a}). In \citep{luo2015fully} and \citep{ni2014comparison} execution
of ranking and selection in a parallel environment is discussed and
\citep{NelsonBoot2014} extends the scope from Normal distributions to general
distributions using a bootstrap approach.

In the present paper, we introduce a new Bayesian procedure for ranking and
selection called \BayesRS that combines some of the features mentioned
above. Observations are assumed to be from a multivariate Normal distribution
with unknown mean and covariance matrix, where we use a so-called uninformative
prior distribution.  In a first stage, $ n_0$ complete observations from all
alternatives are made. Then we continue sampling in a sequential fashion until
we are sure that the PCS $ \ge 1-\alpha$ for a given $ \alpha$. For each iteration we
are given a fixed computing budget of $ b$ simulations that has to be allocated to the
alternatives. The only input parameters therefore are $ n_0, \alpha$ and $b$.  

Within our procedure \BayesRS we compare two different allocation
strategies. The first, \GreedyOCBA, is adapted from the optimal computing budget
allocation strategy of \citep{Chen1997}, see also \citep{Branke2007}. It
allocates the budget according to the increase in the PCS we can expect, if the
whole budget would be given to a single alternative. Our new strategy \DPW uses
the \emph{dominance probability} of a pair $ (i,j)$ of alternatives, i.e.\ the
posterior probability that the mean of alternative $ i$ is less than the mean of
$ j$. \DPW allocates simulations such that the dominance probabilities of all
pairs relevant for the present ranking and selection task (see below) are
increased.  As these dominance probabilities were also used by \BayesRS for the Bonferroni lower bound
of the PCS in the last stage, \DPW may simply re-use these values and allocate
the simulations for the present stage. 

To determine the dominance probabilities, we need the posterior distribution of
the unknown mean. Due to a possibly unequal allocation of the simulation budget
we may have incomplete observation for the different alternatives. Our
particular sampling scheme (see Section \ref{subsec:IterativeAllocation})
guarantees a so-called monotone pattern of missing data, i.e.\ observations are
missing only at the end of the sample. For this case the posterior distribution
of the means is known in case the covariance matrix $ \SIGMA$ is given. No such
result seems to exist for unknown $ \SIGMA$. We therefore use an approximation
that is adapted from the well-studied case of complete observations. Our
empirical experiments underline the practicability of this approximation, the
errors seem to be well below the bounds we imposed.

Many ranking and selection procedures only work for the selection of the best
alternative with minimal (or maximal) mean. Our procedure also works for more
general targets as long as they can be defined by pairwise comparison of
alternatives, as e.g.\ determining the $ m$ best alternatives and rank them or
determine the alternative with median mean. 

In our empirical tests we compare the two allocation strategies \GreedyOCBA and
\DPW in different scenarios. We also compared \BayesRS with allocation \DPW to
two procedures well-known from the literature, namely \KN from
\cite{kim2006asymptotic} and \PLUCK from \cite{qu2012ranking}(see also
\citep{qu2015sequential}). In particular \PLUCK seems to be similar to our
set-up as it is a fully sequential procedure that allows dependent observations
and requires only mean and scale matrix of the prior distribution to be known.
In our experiments however, our procedure \BayesRS seemed to be far superior to
both, \KN and \PLUCK.

The main contributions of the present paper are: it introduces a new
R\&S-procedure that allows for dependent sampling without any prior knowledge of
parameters, it uses a general target scheme for the alternatives to be selected
and it introduces a new simple, and empirically very efficient allocation rule
\DPW that is based on a new approximation of the posterior distributions.

This paper is based on the doctoral thesis \cite{Gorder2012}.  It is organized
as follows. In Section \ref{sec:MathModel} we give the exact mathematical
description of the sampling process and our Bayesian model. Technical details
about the posterior distribution with missing data are sketched in an
Appendix. Section \ref{sec:RankSelect} generalizes the concept of ranking and
selection of solutions and gives a simple Bonferroni bound for the PCS. A precise
definition of our complete ranking and selection algorithm is given in Section
\ref{sec:TheAlgorithm}. In Section \ref{sec:Alloc} we introduce the new
allocation rule \DPW and the adapted OCBA procedure \GreedyOCBA.  We report on
extensive empirical tests of our algorithm and its comparison to \KN and \PLUCK
in Section \ref{sec:Computational-study}. Some conclusions are given in the
final Section \ref{sec:Conclusion}.

\section{The Mathematical Model}\label{sec:MathModel}
\subsection{A Bayesian Environment}
Let $ \L:=\{1 ,\ldots, L\}$ denote the fixed set of alternatives.  $ X_{ik}$ is the
$ k$-th observation of alternative $ i\in \L$. 
Simulating \emph{all} of the alternatives from $ \L$ with the $ k$-th scenario
therefore leads to a (column) vector of observations
\begin{equation}\label{eq:DefXPj}
\X_{\Pkt k}:=(X_{1k} ,\ldots, X_{Lk})^T,\quad k=1,2\ldots.
\end{equation}
We assume that for known $ \MU, \SIGMA$, the sequence $ \X_{\Pkt 1},\X_{\Pkt 2}
,\ldots$ of observations are independent and identically $
\mathcal{N}_L(\MU,\SIGMA)$-distributed. Here, $\mathcal{N}_L(\MU,\SIGMA)$
denotes the $ L$-dimen\-sional normal distribution with mean $ \MU=(\mu_1
,\ldots, \mu_L)^T$ $\in \R^L$ and the positive definite $ L\times L$ covariance
matrix $ \SIGMA$.  This model includes the case, where the $ L$ alternatives are
simulated independently (i.e.\ with different scenarios), then $ \SIGMA$ is a
diagonal matrix. $ \MU$ and $ \SIGMA$ are assumed to be unknown and we want to
extract information about $ \MU$ from the observations $ \X$. As we assume that
the simulations are performed on a computer, we may identify the $ k$-th
scenario with $ z_k$, the $ k$-th seed for the random generator of the
observations.
 
We take a Bayesian point of view and assume that the unknown parameters
$ \MU$ and $ \SIGMA$ are themselves observations of random variables $ W$ and $
S$ having a prior distribution with density $ \pi(\MU,\SIGMA)$. We do not assume
any specific prior knowledge about the parameters, therefore we shall use the
so-called non-informative prior distribution (see \cite{DeGroot2004d}) with
\begin{equation}\label{eq:uninformed}
\pi(\MU,\SIGMA)\,\propto\,\det(\SIGMA)^{-\frac{\nu_0+L+1}{2}}
\end{equation}
where $ \propto$ means that the right hand side gives the density $ \pi$ up
to some multiplicative constant that does not depend on $ \MU$ or $ \SIGMA$. $
\nu_0$ is a so-called hyper-parameter that allows to control the degree of
uncertainty about $ \SIGMA$ and is set to $ L-1$ in our experiments.

\subsection{Iterative allocation of simulation runs}
\label{subsec:IterativeAllocation}

In the first stage, all alternatives are simulated $ n_0$ times for a fixed
number $ n_0$. From iteration $ n_0+1$ on, the simulation budget $ b$ is
allocated sequentially to the alternatives, depending on the observations made
so far.%, see the detailed algorithm in Section~\ref{sec:TheAlgorithm} below.
This may result in samples with \emph{missing data}. 

We assume the following particular CRN sampling scheme: if in iteration $ k>n_0$
simulations are allocated to alternative $ i$, then these simulations start with
scenario (seed) $z_{l_0}$, where $ l_0$ is the smallest number $ l$ such that $
z_l$ has not yet been used for alternative $ i$, in other words, $z_{l_0-1}$ is
the last scenario used for alternative $ i$, see the example in Figure
\ref{fig:SamplingScheme}.

\begin{figure}[htb]
  \centering
  \includegraphics[width=7cm]{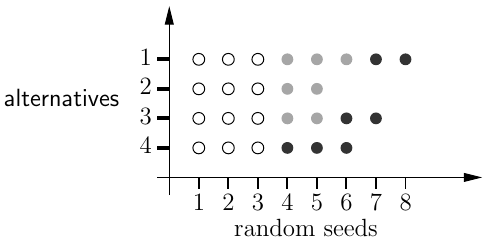}
  \caption{The first five iterations of a sampling scheme with $ L=4, n_0=3$
    and budget $ b=7$. The first $ n_0=3$ iteration use complete observations, the
    allocation of the $ b$ simulations of the forth iteration are
    marked in gray, those of the fifth in black. $ 8$ different scenarios (or
    seeds) are in use.  }
  \label{fig:SamplingScheme}
\end{figure}

This results in samples that may contain different numbers of observations for
different alternatives, but values may be missing only at the end of the sample,
a so-called \emph{monotone} pattern of missing data. As this pattern does not
depend on the unobserved data (but only on the observed ones via the allocation
rule) it is also \emph{ignorable}, see \cite{schafer1997analysis} for an exact
definition of these terms. Although this may seem rather artificial, it is
crucial for the distribution analysis below. It is easily implemented by a
suitable book-keeping of the random seeds.

We use
\begin{equation}\label{eq:DefXiP}
 \X_{i\Pkt}:=(X_{i1},X_{i2} ,\ldots, X_{in_i})
\end{equation}
to denote the random (row) vector of the $ n_i$ (consecutive) observations
produced for the $ i$-th alternative, $ i=1 ,\ldots, L$, under this scheme and,
similarly, $ \x_{i\Pkt}:=(x_{i1},x_{i2} ,\ldots,$ $x_{in_i})$ for a specific
sample.

Let $ \n=(n_1 ,\ldots, n_L)\in \N^L$ be the vector of the present sample sizes
for each of the $ L$ alternatives, then the variables $ \X_{i\Pkt}, i=1 ,\ldots,
L,$   can be collected into
a matrix-like scheme with possibly different row lengths $ n_i$ :
\begin{equation}\label{eq:Xn}
\X_{(\n)}:= \X_{(n_1 ,\ldots, n_L)}:=
\begin{pmatrix}
   X_{11},& X_{12},&\ldots &X_{1n_1}\\
  X_{21},& X_{22},&\ldots &X_{2n_2}\\
\vdots & \vdots& & \vdots\\
  X_{L1},& X_{L2},&\ldots &X_{Ln_L} 
\end{pmatrix}.
\end{equation}
Columns $ \X_{\Pkt k}$ may be incomplete for $ k>n_0$, as the $ k$-th scenario
or seed may not have been allocated to all alternatives. In the example in
Figure~\ref{fig:SamplingScheme}, columns $6,7,8$ are incomplete.  Let $
\R^{L\otimes \n}$ denote the set of possible samples $ \x_{(\n)}$ that may be
observed with $ \X_{(\n)}$ for a particular size vector $ \n=(n_1 ,\ldots,
n_L)$. An \emph{allocation rule} is a mapping $Q: \R^{L\otimes \n} \to \N^L$
that determines the numbers $ q_1 ,\ldots, q_L$ of additional simulations for
each alternative. Then $ \n$ is updated to $ \n':=(n_1+q_1 ,\ldots, n_L+q_L)$
and the new observations are added at the end of each line of $ \x_{(n)}$ to
form the new sample $ \x_{(\n')}$, an element of $ \R^{L\otimes \n'}$.

Note, that with this sampling scheme, the simulations in different iterations
need not be independent as for some alternatives we may have to \emph{re-use}
scenarios that have already been used in earlier iterations for other
alternatives. E.g.\ in Figure \ref{fig:SamplingScheme}, the next simulation with
alternative $2$ would have to use seed $ z_6$.  This means that we have to keep
random seeds until all alternatives have been simulated with this seed. A new
seed resulting in an independent observation is used for the simulation of some
alternative $ i$ only if $ n_i=\max\{n_1 ,\ldots, n_L\}$, that is all seeds used
so far have been applied to $ i$. This would be the case for alternative $ 1$ in
Figure~\ref{fig:SamplingScheme}.

The allocation schemes we use below depend on the posterior distribution of the
mean $ W$ given data $ \x\in  \R^{L\otimes \n}$ which is determined in the next Subsection.

\subsection{Likelihood and posterior distributions with missing data}\label{subsec:MissingData}

The likelihood function of incomplete samples as described in the last section
are examined in great detail and generality e.g.\ in \cite{schafer1997analysis} and
\cite{Dominici2000}. Our case is comparatively simple, as our sampling scheme
guarantees monotone and ignorable patterns of missing data.

To simplify notation, let the set of alternatives $ \L=\{1 ,\ldots, L\}$ be ordered
such that for the present data $ \x =\x_{(\n)}\in  \R^{L\otimes \n}$ with $ \n=(n_1 ,\ldots,
n_L)$ we have
\begin{equation}\label{eq:Nordered}
n_1\ge n_2\ge \cdots \ge n_L.
\end{equation}

\begin{figure}[tb]
  \centering
  \includegraphics[width=8cm]{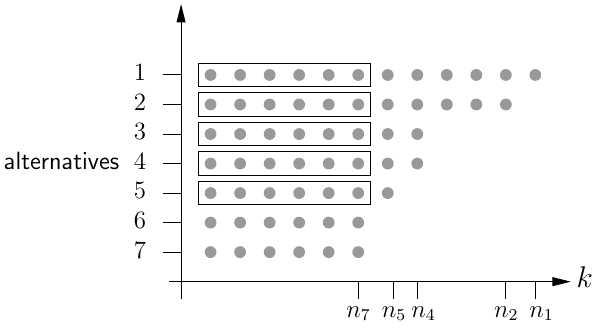}
  \caption{A sample with $ \L=\{1 ,\ldots, 7\}$ and $ n_1=12, n_2=11, n_3=n_4=8,
    n_5=7, n_6=n_7=6$. $\xx^{(n_6)}_{[<6]} $ contains all $ 5$ sample means from
    the boxed values. }
\label{fig:MonotoneSample}
\end{figure}

We need projections of vectors and matrices to components corresponding to
subsets of the alternatives $ \L=\{1 ,\ldots, L\}$.  For $ \y=(y_1 ,\ldots, y_L)^T\in \R^L$ and $ i\in \{1
,\ldots, L\}$ let
\begin{equation}\label{eq:DefProjVector}
  \y_{[<i]}:= (y_1 ,\ldots, y_{i-1})^T\in \R^{i-1},\quad \text{ and }\quad  \y_{[\le i]}:= (y_1 ,\ldots, y_{i})^T\in \R^{i}.
\end{equation}
Similarly, for the $ L\times L$ matrix $ \SIGMA=(\sigma_{k,l})_{k,l = 1 ,\ldots,
  L}$ we define
\begin{eqnarray}
  \SIGMA_{[\le i]}&:=&(\sigma_{kl})_{k,l = 1 ,\ldots, i}, \qquad  \SIGMA_{[<i]}:=(\sigma_{kl})_{k,l = 1 ,\ldots, i-1}, \notag\\
  \SIGMA_{[i,<i]}&:=& (\sigma_{i1} ,\ldots, \sigma_{i,i-1})\quad \text{ and}\quad \beta_i := \Sigma_{[i,<i]}\big(\Sigma_{[<i]}\big)^{-1} \label{eq:DefBeta}.
 \end{eqnarray}
%Here, $ \beta_i$ is a $ (i-1)$-dimensional row vector.} 
 For $ 1\le i\le k\le L$ we have $ n_i\ge n_k$ and we may define the sample mean
 of alternative $ i$ restricted to the first $ n_k$ observations
\begin{equation}\label{eq:DefMean}
\xx^{(n_k)}_i:=\frac{1}{n_k} \sum_{j=1}^{n_k}x_{ij} .
\end{equation}
Then ${\xx}_{i}:= \xx^{(n_i)}_i$ is the sample mean of alternative $ i$ over all
$ n_i$ observations.  See Fig.~\ref{fig:MonotoneSample} for an illustration of $
\xx^{(n_i)}_{[<i]}$.

For the dominance probabilities we need the posterior distribution of $ W$ given
an incomplete observation $ \x$. We first assume, that the covariance matrix is
known and that the uninformative prior $ \pi(\mu)\propto 1$ for the mean is used.

%%%%%%%%%%%%%%%%%%%%%%%%%%%%%%%%%%%%%%%%%%%%%%%%%%%%%%%%%%%%%%%%%%%%%%%%%%%%%%%%%%%%
\begin{Theorem}\label{the:posterior}
  Let the conditional distribution of (the columns) $ \X_{\Pkt k}, \ k=1,2
  ,\ldots$, be \iid $ \NN_L(\MU,\SIGMA)$ distributed given $ W=\MU$ and $
  S=\SIGMA$. Assume that $ W$ has the non-informative prior $ \pi(\mu)\propto
  1$. Let the (possibly incomplete) data $ \x=\x_{(\n)}\in \R^{L\otimes \n}$
  with $ \n=(n_1 ,\ldots, n_L),\ n_1\ge n_2\ge \cdots \ge n_L$ be given as
  described in Section~\ref{subsec:IterativeAllocation}.

  Then the posterior distribution of $ W$ given $ \X=\x, S=\SIGMA$ is an $ L$-dimensional
  Normal distribution $ \mathcal{N}_L(\NU,\LAMBDA)$ with mean $\NU=(\nu_1
  ,\ldots, \nu_L) $ where
  \begin{equation}\label{eq:DefNuTheorem}
\nu_1 := {\xx}_{1}, \qquad
 \nu_i:= {\xx}_i +\beta_i( \NU_{[<i]} - {\xx}_{[<i]}^{(n_i)})\quad \text{ for } i=2 ,\ldots, L,
\end{equation}
and   covariance matrix $ \LAMBDA := \LAMBDA_{[\le L]}$ where  
\begin{equation}\label{eq:DefLambdaTheorem}
  \LAMBDA_{[\le 1]}=\frac{\sigma_{11}}{n_1},\qquad \LAMBDA_{[\le i]}= \left(\begin{array}{c|l}
      \LAMBDA_{[<i]} &\LAMBDA_{[<i]}\ \beta_i^T \\[0.9ex]\hline
      \rule[0.5ex]{0pt}{2ex}\beta_i\ \LAMBDA_{[<i]} & \frac{\sigma_{ii}}{n_i} + \beta_i(\LAMBDA_{[<i]}-\frac{1}{n_i}\SIGMA_{[<i]})\beta_i^T
\end{array}\right)
\end{equation}
 for $ i=2 ,\ldots, L$.
\end{Theorem}
%%%%%%%%%%%%%%%%%%%%%%%%%%%%%%%%%%%%%%%%%%%%%%%%%%%%%%%%%%%%%%%%%%%%%%
A proof of Theorem \ref{the:posterior} is essentially contained in
\cite{schafer1997analysis}, it is sketched in the Appendix.

For the case of an \emph{unknown} covariance matrix $ \SIGMA$ and prior
distribution as in \eqref{eq:uninformed}, a factorization of the likelihood is
again given in \cite{schafer1997analysis}, but only for a complex
parameterization that is described in the Appendix. There seems to be no way to
obtain a closed expression for the posterior distribution of $ W$ in this
case. We therefore use the results of Theorem \ref{the:posterior} and replace
the unknown $ \SIGMA$ by its estimate, adapting parameters similarly as it is done in the case of
complete observations.

For the estimation of $ \SIGMA_{[<i]}=\big(\sigma_{kl}\big)_{1\le k,l<i}$ and $
\beta_i$ we only consider alternatives $ 1 ,\ldots, i-1$ that all have sample
sizes $ \ge n_{i-1}$. We may therefore use the (maximum likelihood) estimates
\begin{equation}\label{eq:DefsigEst}
  \hat{\sigma}^{(n_{i-1})}_{kl} =  \hat{\sigma}^{(n_{i-1})}_{kl}(\x):= \frac{1}{n_{i-1}}\sum_{m=1}^{n_{i-1}}(x_{km} -{\xx}^{(n_{i-1})}_k)(x_{lm}-{\xx}_l^{(n_{i-1})})
\end{equation}
for $ 1\le k,l<i$ and $ \x=\x_{(\n)}$ and put
\begin{eqnarray}
 \hat{\SIGMA}_{[< i]}= \hat{\SIGMA}_{[< i]}(\x)&:=&\Big( \hat{\sigma}^{(n_{i-1})}_{kl} \Big)_{k,l=1 ,\ldots, i-1} \label{eq:DefSIGEst}\\
\hat{\SIGMA}_{[i,<i]}=\hat{\SIGMA}_{[i,<i]}(\x)&:=& (\hat{\sigma}^{(n_{i-1})}_{i1} ,\ldots, \hat{\sigma}^{(n_{i-1})}_{i,i-1}),\notag\\
\hat{\beta}_i=\hat{\beta}_i(\x)&:=& \hat{\SIGMA}_{[i,<i]}(\hat{\SIGMA}_{[< i]})^{-1}.\label{eq:DefBetaEst}
\end{eqnarray}
Note that $ \hat{\SIGMA}_{[< i]}$ is not necessarily contained in $
\hat{\SIGMA}_{[\le i]}$ as the estimates use possibly different sample sizes $
n_{i-1}$ and $ n_i$.  In \eqref{eq:DefBetaEst}, we have to make sure that
$\hat{\SIGMA}_{[< i]} $ is nonsingular. From \cite{Dykstra1970} it is known that
if the sample size $ n_{i-1}$ of $\hat{\SIGMA}_{[< i]}$ fulfills $n_{i-1}> i-1$,
then $\hat{\SIGMA}_{[< i]}$ is positive definite with probability one.  This
could be guaranteed, if we require for the initial sample size $ n_0\ge L $ as
then $n_{i-1} \ge n_0 \ge L > i-1$.

We  now plug these estimates into the definition of the posterior means and
obtain
\begin{equation}\label{eq:NuHat}
  \begin{split}
     &\hat{\nu}_1=\hat{\nu}_1(\x):= {\xx}_{1}, \quad 
  \hat{\nu}_i=\hat{\nu}_i(\x):=  {\xx}_i +\hat{\beta}_i( \hat{\NU}_{[<i]} -
  {\xx}_{[<i]}^{(n_i)})\quad \text{ and } \\
  &\hat{\NU}(\x):=(\hat{\nu}_1 ,\ldots, \hat{\nu}_L).
\end{split}
\end{equation}
The case of $\hat{ \LAMBDA}$ is more complicated. To obtain an adequate
estimate $\hat{ \LAMBDA}$, we first look at the case of \emph{complete}
observations, i.e.\ with $ n=n_1 =\cdots = n_L>L$. Then it is well-known (see
e.g. \cite{DeGroot2004d}, 10.3) that with $ \SIGMA$ \emph{known} and an
non-informative prior distribution for the mean $ W$, the posterior distribution
of $ W$ would be Normal with 
\begin{equation}\label{eq:ComplSigmaKnown}
\text{mean $ \xx$ and  covariance matrix $ \frac{1}{n}\SIGMA$.}
\end{equation}
  If $
\SIGMA$ is \emph{unknown} with the non-informative prior distribution as in
\eqref{eq:uninformed}, the marginal posterior distribution of the mean $ W$ in
the complete observation case is an $ L$-dimensional $ \t$-distribution with $
n-L+\nu_0$ degrees of freedom, 
\begin{equation}\label{eq:ComplSigmaUnKnown}
\text{location parameter $ \xx$ and scale matrix
$\frac{1}{ n-L+\nu_0}\hat{\SIGMA}_{[\le L]}(\x)$}
\end{equation}
(see \cite{DeGroot2004d},10.3).

Switching to incomplete observations with \emph{known} covariance matrix,
Theorem \ref{the:posterior} tells us that the posterior mean $ \xx$ of
\eqref{eq:ComplSigmaKnown} has to be replaced by $ \NU$ and the covariance $
\frac{1}{n}\SIGMA$ of \eqref{eq:ComplSigmaKnown} has to be replaced by $
\LAMBDA$ as in \eqref{eq:DefLambdaTheorem}, reflecting the different sample
sizes for each alternative. Note, that for $n= n_1=\cdots=n_L$, we have $
\NU=\xx$ and $ \LAMBDA =\frac{1}{n}\SIGMA$. In the case of \emph{unknown} $
\SIGMA$, it therefore seems reasonable to approximate the posterior distribution
of the mean $ W$ by an $ L$-dimensional $ \t$-distribution with $ n_L-L+\nu_0$
degrees of freedom, location parameter $ \hat{\NU}(\x)$ and a suitable scale
matrix $ \hat{\LAMBDA}$. The scale matrix $ \hat{\LAMBDA}$ should be obtained
from $ \LAMBDA$ with the $ \sigma_{kl}$ replaced by their estimates in a similar
fashion as $\frac{1}{n} \SIGMA$ in \eqref{eq:ComplSigmaKnown} is replaced by $
\frac{1}{n-L+\nu_0}\hat{\SIGMA}$ in \eqref{eq:ComplSigmaUnKnown}. In particular, in the $
i$-th iteration of the recursive definition of $ \LAMBDA$ in
\eqref{eq:DefLambdaTheorem}, the constant factor $\frac{1}{ n-L+\nu_0}$ from
\eqref{eq:ComplSigmaUnKnown} should be replaced by $ \frac{1}{n_i-L+\nu_0}$ as
in this step the sample size $ n_i\ge n_L$ is used.

We therefore choose as approximation to the posterior covariance matrix
\begin{align}
  \hat{\LAMBDA}_{[\le 1]}=\hat{\LAMBDA}_{[\le 1]}(\x)&:=\frac{\hat{\sigma}^{(n_1)}_{11}}{n_1-L+\nu_0},\label{eq:LAMBDAHat}  \\
  \hat{\LAMBDA}_{[\le i]}= \hat{\LAMBDA}_{[\le i]}(\x)
  &:= \left(\begin{array}{c|l}
      \hat{\LAMBDA}_{[<i]} &\hat{\LAMBDA}_{[<i]}\ \hat{\beta}_i^T \\[0.9ex]\hline
      \rule[0.5ex]{0pt}{2ex}\hat{\beta}_i\ \hat{\LAMBDA}_{[<i]} &
      \frac{\hat{\sigma}^{(n_i)}_{ii}}{n_i-L+\nu_0}+  \hat{\beta}_i\ \big(
      \hat{\LAMBDA}_{[<i]} - \frac{1}{n_i-L+\nu_0}\hat{\SIGMA}_{[<i]}\big)\hat{\beta}_i^T
\end{array}\right) \notag
\end{align}
for $ i=2 ,\ldots, L$, and $
\hat{\LAMBDA}=\hat{\LAMBDA}(\x):=\hat{\LAMBDA}_{[\le L]}(\x)$.  As an additional
justification for that choice of $ \hat{\LAMBDA}$ we may add that for the
complete observation case with $ n:=n_1=\cdots=n_L$, our matrix $ \hat{\LAMBDA}$
is equal to the scale matrix $ \frac{1}{n-L+\nu_0}\hat{\SIGMA}_{[\le L]}$ as in
\eqref{eq:ComplSigmaUnKnown}, thus $ \hat{\LAMBDA}$ generalizes the complete
observation case to our situation with incomplete observations.
In the sequel, we therefore assume that the posterior distribution of $ W$ given
$ \X=\x$ is a $ L$-dimensional $ \t$-distribution with     $ n_L-L+\nu_0$ degrees of
freedom, location parameter $ \hat{\NU}(\x)$ and scale matrix $
\hat{\LAMBDA}$.

The posterior distribution of the mean $ W$ is used here to determine the
\emph{dominance probability} $ p_{ij}^\delta $ of pairs $ (i,j)$ of
alternatives, this is the posterior probability that alternative $ i$ is better
than $ j$, i.e.\ the posterior probability  of the event $ W_i\le
W_j+\delta$ for an \emph{indifference zone parameter} $ \delta\ge 0$:
  \begin{equation}\label{eq:DefDP}
  p_{ij}^\delta:=\P[W_i\le W_j+\delta\mid \X=\x]\ =\ \P[W_i-W_j\le \delta\mid \X=\x].
  \end{equation}
Let $ G(\cdot\ ;k,a,b)$ denote the distribution function of the one-dimensional $
t$-distribution with $ k$ degrees of freedom, location parameter $ a$ and scale
parameter $ b$.  From the discussion above and standard 
  properties of the multivariate $ \t$-distribution we may then conclude that
  \begin{equation}\label{eq:DP}
p_{ij}^\delta \approx \ G(\delta;\ n_L-L+\nu_0,\; \hat{\nu}_i-\hat{\nu}_j,\;
  \hat{\Lambda}_{ii}+\hat{\Lambda}_{jj}-2\hat{\Lambda}_{ij}).
\end{equation}
The dominance probability $ p_{ij}^\delta$ of alternative $i $
over alternative $ j$  allows to
lower bound the PCS for ranking and selection targets that are based on pairwise
comparison as they are introduced now.

\section{General Ranking and  Selection Schemes} \label{sec:RankSelect}
\subsection{Target and Selection}\label{subsec:TargetAndSelection}
Our ranking and selection scheme is a generalization of the approach in
\cite{Schmidt2006}.  We want to select alternatives from the set $\L= \{1
,\ldots, L\}$ according to their ranks under some performance measure, in our
case the (estimated) mean value. E.g.\ we want to select the $ m$ alternatives with the
lowest mean values and rank them as it is required in ant algorithms.  We
restrict ourselves to such selections that can be determined using pairwise
comparisons of alternatives and then apply \eqref{eq:DefDP} to bound the error
probability.  We use an abstract concept which is illustrated by some examples
below.

We first define a set $ A\subset \{1 ,\ldots, L\}$ of \emph{target ranks}. We
want to select those alternatives that have ranks from $ A$ with respect to
their estimated mean values. We also determine whether the selected alternatives
should be ranked according to these values.

We estimate the unknown means by the present posterior means $ \hat{\nu}_1(\x)
,\ldots$, $ \hat{\nu}_L(\x)$ and order them as
\begin{equation}\label{eq:OrderedPMeans}
 \hat{\nu}_{i_1}(\x)<\hat{\nu}_{i_2}(\x) < \cdots <
\hat{\nu}_{i_L}(\x).
\end{equation}
This is possible with probability one as we have continuous posterior
distributions for which $\P[W_i=W_j \text{ for some } i\not=j\mid \X=\x] =
0$. We then estimate the ranks of the alternatives by their ranks in
\eqref{eq:OrderedPMeans} and select those alternatives that have estimated ranks
in the target set $ A$, i.e.\ we select the
alternatives
\begin{equation}\label{eq:DefSel}
B=B(\x):=\{i_j\mid j\in A\},
\end{equation}
where the $ i_j$ are taken from \eqref{eq:OrderedPMeans}. The probability that
this is a correct selection is the posterior probability that the actual ranks
of the selected means $ W_l, l\in B$, are those required by the target set $ A$,
i.e.
\begin{equation}\label{eq:DefPCS}
PCS:=\P[\,\{\text{rank}_{W}(W_l)\mid l \in B\}=A \mid \X=\x\,]
\end{equation}
where $ \text{rank}_{\t}(t_j)$ denotes the rank of $ t_j$ within $ \t=(t_1
,\ldots, t_L)\in \R^L$. If it is required that the selected alternatives are
also ranked among themselves then  \eqref{eq:DefPCS} is
replaced by
 \begin{equation}\label{eq:DefPCSRanked}
   \begin{split}     
PCS:=\P\big[\,\{\text{rank}_{W}(W_l)&\mid l \in B\}=A  \text{ and } \\
& W_j\le W_l \text{ for } j,l\in B, j<l \ \mid \X=\x\,\big]
\end{split}
\end{equation}
We restrict ourselves here to target sets that can be described  by
pairwise comparisons, i.e.\ we assume that there is  a set $\rho_{AB}$ of pairs from  $
\L=\{1 ,\ldots, L\}$ (or, more formally, a  binary relation over $\L\times \L$)
such that 
\begin{equation}\label{eq:RhoCorrect}
  \begin{split}
    \{\text{rank}_{W}(W_l)&\mid l \in B\}=A \\
    &\iff W_i\le W_j \ \text{ for all } (i,j) \in \rho_{AB}.
  \end{split}
\end{equation}
Then \eqref{eq:DefPCS} becomes
\begin{equation}\label{eq:PCSRho}
PCS = \P\big[W_i\le W_j \ \text{ for all } (i,j) \in \rho_{AB}\mid \X=\x\, \big].
\end{equation}
If an additional ranking is required then $ \rho_{AB}$ must be extended by pairs
$ (j,l)$ with $ j,l\in B$ and $ j<l$, see \eqref{eq:DefPCSRanked}.  The
following examples show how these relations $ \varrho_{AB}$ may be obtained.
\begin{Example} Assume that \eqref{eq:OrderedPMeans} holds.
  \begin{enumerate}\renewcommand{\theenumi}{\textbf{\alph{enumi}}}\renewcommand{\labelenumi}{\theenumi)}
  \item If we want to select the alternative with smallest mean we put $
    A:=\{1\}$, $ B:=\{i_1\}$, and $\rho_{AB}:=\{(i_1,j)\mid j\in \L-\{i_1\}\}$.
Then we have 
\begin{align*}  
\{\text{rank}_{W}(W_l)\mid l \in B\}=A & \iff W_{i_1} \text{ has rank $ 1$ in
} (W_1 ,\ldots, W_L)\\
&\iff W_{i_1} \le W_j\ \text{ for all }
j\not=i_1, j=1 ,\ldots, L\\
&\iff W_i\le W_j \ \text{ for all } (i,j)\in \rho_{AB}.  
\end{align*}
\item If  we want to select the $ m$ best (minimal) alternatives for some $
  m\le L$, we put $A:= \{1 ,\ldots, m\}$ and $ B:=\{i_1 ,\ldots, i_m\}$. As
  characterizing relation we obtain $ \rho_{AB}:=\{(l,j)\mid l\in
  B, j\in \L-B\} $. Then
\begin{align*}
  \{\text{rank}_{W}(W_l)&\mid l \in B\}=A \\
&\iff W_{i_1} ,\ldots, W_{i_m} \text{
 have ranks $ 1 ,\ldots, m$ in } W_1 ,\ldots, W_L\\
 &\iff W_{i_k} \le W_j \ \text{ for all } k=1 ,\ldots, m, j\notin\{i_1 ,\ldots,
 i_m\}\\
 &\iff W_i\le W_j \ \text{ for all } (i,j)\in \rho_{AB}   
\end{align*}
\item If, in addition, the $ m$ best alternatives have to be ranked among
  themselves then we would choose
  \[ 
\rho_{AB}:=\{(i_1,i_2) ,(i_2,i_3)\ldots, (i_{m-1},i_m)\}\cup \{(i_m,j)\mid j\notin B\}.
  \]
This includes the case where a complete ranking of the alternatives is required where
\[ 
\rho_{AB}:=\{(i_1,i_2) ,(i_2,i_3)\ldots, (i_{L-1},i_L)\}.
\]
\item In a similar fashion $ A,B$ and $ \rho_{AB}$ may be defined to select the
  median or the span of $ (\mu_1 ,\ldots, \mu_L)$.   
\end{enumerate}
\end{Example}
In the iterative procedure below, the posterior means $ \hat{\nu}_1(\x) ,\ldots,
\hat{\nu}_L(\x)$ have to be determined after each iteration based on the new
observations. Therefore, $ B$ and $\rho_{AB}$ have to be re-calculated in each
iteration also.

\subsection{A Bound for the Probability of a Correct Selection}

Based on the characterizing relation $ \rho_{AB}$ we may now derive a simple
lower bound of Bonferroni type for the PCS defined in \eqref{eq:PCSRho} with an
indifference parameter $ \delta\ge 0$   as follows
\begin{eqnarray}
PCS^\delta &=& \P\Big[ W_i \le W_j+\delta \text{ for all } (i,j)\in \rho_{AB}\mid \X=\x\Big]\notag\\
%&=& 1- \P\Big[  W_i > W_j+\delta \text{ for some } (i,j) \in  \rho_{AB}\mid \X=\x\Big]\notag\\
&\ge&  1- \sum_{(i,j)\in \rho_{AB}} \P\big[  W_i > W_j+\delta \mid \X=\x\,\big]\notag\\
%&=& 1- \sum_{(i,j)\in \rho_{AB}} \big(1-\P[W_i-W_j\le \delta \mid \X=\x]\big)
\notag\\
&=& 1- \sum_{(i,j)\in \rho_{AB}} \big(1-p_{ij}^{\delta}\big)\ =: LB^\delta(\x). \label{eq:Bonferroni}
 \end{eqnarray}
 where $ p_{ij}^\delta$ was defined in \eqref{eq:DefDP}.  Note that for these
 error bounds we need the dominance probabilities for pairs $ (i,j)\in
 \rho_{AB}$ only.

We now proceed to define our algorithm in full detail.

\section{The Algorithm BayesRS}\label{sec:TheAlgorithm}

Let the following items  be given:

$ \L=\{1 ,\ldots, L\}$ is the set of alternatives, $ A\subset \{1 ,\ldots, L\}$ is
the target set of ranks to be selected and possibly ranked, $ \alpha\in (0,1)$ is the bound for the
error probability, $ b\in \N$ is the simulation budget for each iteration and $
n_0\in \N$ is an initial sample size.{\leftmargin0.5cm}
\begin{description}%{\leftmargin0.5cm}
\item [\textbf{Initialization}]: \ Observe $ \X_{\Pkt j}$ for $ j=1 ,\ldots,
  n_0,$ i.e.\ make $ n_0$ complete observations and let $ \x$ denote the result.

  Determine the posterior means $ \hat{\nu}_j(\x), j\in \L,$ as in
  \eqref{eq:NuHat}, the selection $ B(\x)$ as in \eqref{eq:DefSel}, relation $
  \rho_{AB}$ as in \eqref{eq:RhoCorrect}  and the dominance
  probabilities $ p^\delta_{ij}, (i,j)\in \rho_{AB}$, as in \eqref{eq:DP}. Finally
  calculate the lower bound $ LB^\delta(\x)$ for the PCS as in \eqref{eq:Bonferroni}.\\[-1ex]
\item[\textbf{while}]  $ LB^\delta(\x)\ < \ 1-\alpha$\quad \textbf{do}
  \begin{itemize}
  \item apply an allocation rule $ Q$ to the sample $ \x$ to determine the number of
    additional simulation runs $ \q=(q_1 ,\ldots, q_L):=Q(\x)$,
  \item  perform $ q_i$ simulations with alternative $ i\in \L$, taking into
    account the common random numbers scheme as described in Section~\ref{subsec:IterativeAllocation},
 let $ \x$ be the extended data including the new simulation results,
  \item update the posterior means $ \hat{\NU}(\x)$, the selection $ B(\x)$,  the relation $ \rho_{AB}$,
    the dominance probabilities $ p^\delta_{ij}, (i,j)\in \rho_{AB}$, and the lower
    bound $ LB^\delta(\x)$.
  \end{itemize}
\item []\textbf{return } the present selection $ B(\x)$. 
\end{description}
In \cite{Gorder2012} it is shown that this algorithm terminates if $ \delta>0$
and if at least one simulation is allocated to each alternative in each iteration.

%%%%%%%%%%%%%%%%%%%%%%%%%%%%%%%%%%%%%%%%%%%%%%%%%%
\section{Allocation strategies}\label{sec:Alloc}

In each of the iterations described in Section~\ref{subsec:IterativeAllocation}
it has to be decided how many simulations should be performed for each
alternative $ i\in \L$. 

We first look at a modified version of the so-called
optimal computing budget allocation (OCBA) strategy introduced by Chen (see for
example \cite{Chen2000,Chen2008,Chen2010}). OCBA allocation strategies try to
allocate a fixed simulation budget $b\in \N$ in such a way that the expected
value of the $PCS$ is maximized. As this 
is a difficult non-linear optimization problem, see e.g.\
\cite{chen2005alternative}, most versions of the OCBA solve a substitute problem
and try to maximize a lower bound as $ LB^\delta(\x)$.  For larger
instances, even this reduced problem can only be solved heuristically, see e.g.
\citep{Chen1996,Chen1997,Branke2007}.

We adapt this approach to our environment with dependent observations. Let $
p_{ij}^\delta(q_i,q_j)$ be an estimate of the dominance probability from
\eqref{eq:DP} with $ q_i$ and $ q_j$ additional simulations for alternatives $
i$ and $ j$. We assume as in \citep{Branke2007} that the current estimates of
the means in $\hat{\nu}_i-\hat{\nu}_j $ will not be affected substantially by
the additional simulations. The degree of freedom $ n_L-L+\nu_0$ of the $
t$-distribution will only change if additional simulations are performed with
alternative $ L$. The effect of additional simulations on the covariance is
difficult to foresee, we use a weighted mixture $ \Gamma_{ij}$ of the current
values $\hat{\Lambda}_{ii},\hat{\Lambda}_{jj},\hat{\Lambda}_{ij} $: 
\begin{align*}
  \Gamma_{ij}:=&\frac{n_{i}}{n_{i}+q_{i}}\hat{\Lambda}_{ii}\,+\,\frac{n_{j}}{n_{j}+q_{j}}\hat{\Lambda}_{jj}\\
  &-2\left[\frac{\hat{\Lambda}_{ii}}{\hat{\Lambda}_{ii}+\hat{\Lambda}_{jj}}\frac{n_{i}}{n_{i}+q_{i}}\,+\,\frac{\hat{\Lambda}_{jj}}{\hat{\Lambda}_{ii}+\hat{\Lambda}_{jj}}\frac{n_{j}}{n_{j}+q_{j}}\right]\hat{\Lambda}_{ij}.
\end{align*}
and put
\begin{align}
  p^\delta_{ij}(q_{i},q_{j}) &:=
\begin{cases}
  G\left(\delta;\,n_L-L+\nu_0\,,\hat{\nu}_{i}-\hat{\nu}_{j},\,\Gamma_{ij}\right)&
  \text{ if } i,j\not= L\\
  G\left(\delta;\,n_L+q_L-L+\nu_0\,,\hat{\nu}_{i}-\hat{\nu}_{j},\,\Gamma_{ij}\right)&
  \text{ if } i= L\text{ or }j=L
\end{cases},
\label{eq:DefpijTilde}\\
\text{  and }&\notag\\ 
\tilde{LB}^\delta(\x;\boldsymbol{q})\,&:=\,1-\sum_{(i,j)\in\rho_{AB}}\Big(1-p^\delta_{ij}(q_{i},q_{j})\Big).\notag
\end{align}
 
We then define a greedy heuristic (\GreedyOCBA) similar to \cite{Chen2010}
and \cite{Branke2007}. It assigns additional simulations to an alternative $ l$
proportional to the  increase in $ \tilde{LB}^\delta$ if the whole budget
was assigned to $ l$. With $\boldsymbol{e}_{l} $ denoting the $ l$-th unit
vector, this increase is given by
\begin{align}
  \Delta_{l} & :=  \tilde{LB}^\delta\left(\x;\, b\cdot\boldsymbol{e}_{l}\right)-\tilde{LB}^\delta(\x,\bm{0})\label{eq.GrdOCBA}\\
  & =
  \sum_{\overset{j\in\CandSet:}{(l,j)\in\rho_{AB}}}\left[p^\delta_{lj}(b,0)-p^\delta_{lj}(0,0)\right]
  +\negthickspace\negthickspace\sum_{\overset{i\in\CandSet:}{(i,l)\in\rho_{AB}}}\left[p^\delta_{il}(0,b)-p^\delta_{il}(0,0)\right],\notag
\end{align}
Note that $p^\delta_{il}(0,0)=
p^\delta_{il}$.  Then, \GreedyOCBA
distributes the simulation budget proportional to the weights $\Delta_{l}$, i.e.
\begin{equation}\label{eq:DefQ}
q_{l}:=\left\lfloor
  b\cdot\frac{\Delta_{l}}{\sum_{i=1}^{L}\Delta_{i}}\right\rfloor ,\quad l=1
,\ldots, L.
\end{equation}
A remaining budget $R:=b-\sum_{i=1}^L q_{i}$ is allocated according to the
largest remainder method. 

To determine the weights $ \Delta_l$, we need the dominance probabilities $
p_{ij}^\delta$ that have been calculated for lower bound $ LB^\delta$  in the
last iteration, but we also need the $2\cdot|\rho_{AB}|$ probabilities $
p^\delta_{ij}(b,0)$ and $p^\delta_{ij}(0,b) , (i,j)\in \rho_{AB}$.
We now introduce a new allocation scheme \DPW (for 'dominance probability
weighting') that works with the $ p_{i,j}^\delta$ of the lower bound $
LB^\delta$ only. 

For $ (i,j)\in \rho_{AB}$, a small value of $ p_{i,j}^\delta$
indicates that the assertion of $ W_i\le W_j+\delta$ needs additional data. We
therefore replace the weights from \eqref{eq.GrdOCBA} by
\begin{equation}\label{eq:DefDeltatilde} 
\tilde{\Delta}_l:=  \max\Big\{
(1-p^\delta_{ij})\frac{\hat{\Lambda}_{ll}}{\hat{\Lambda}_{ii}+\hat{\Lambda}_{jj}}\ \big|\ (i,j)\in
\rho_{AB},\ i=l \text{ or }l=j \Big\},
\end{equation}
Here, those alternatives $ l$ get a larger weight that (a) are part of a pair $
(i,j)\in \rho_{AB}$ with a small dominance probability $ p_{ij}^\delta$ and (b)
have a relatively large estimated variance $\hat{\Lambda}_{ll}$ compared to
their partners, such that
$\hat{\Lambda}_{ll}/(\hat{\Lambda}_{ii}+\hat{\Lambda}_{jj}) $ tends to be
large. In this case, additional simulations with $ l$ might decrease the
variance and increase the certainty of the dominance relation $ (i,j)$. We may
even restrict \eqref{eq:DefDeltatilde} to pairs $ (i,j)\in \rho_{AB}$ with dominance
probabilities
\[ 
p^\delta_{ij} < 1-\frac{\alpha}{|\rho_{AB}|}
\]
where $1- \alpha$ is the given lower bound required for the PCS. Hence we drop
those pairs $ (i,j)$ from \eqref{eq:DefDeltatilde}, that  give an above average
contribution to $ LB^\delta>1-\alpha$ already.  \DPW then allocates the budget proportional to the weights $
\tilde{\Delta}_1 ,\ldots,\tilde{\Delta}_L $ as in \eqref{eq:DefQ}.

In the empirical tests given in the next Section, \DPW performs slightly better than
\GreedyOCBA on the average though \GreedyOCBA is computationally more complex
and uses the information from $ 2|\rho_{AB}|$ additional dominance probabilities.

\section{Computational study}\label{sec:Computational-study}

We implemented the algorithm \BayesRS and compared its efficiency to other R\&S
procedures.  The main objectives of the study were:
\begin{itemize}
\item to compare different allocation strategies within the framework of\\
  \BayesRS, namely \GreedyOCBA, \DPW and \EqAlloc, the naive
  equal allocation of the budget to alternatives,
\item to show the efficiency of \BayesRS and \DPW for different
  correlations of the observations,
\item to compare \BayesRS with  \DPW to other
  R\&S-strategies from literature, namely \KN, see \cite{kim2006asymptotic}, and
  \PLUCK from \citep{qu2012ranking}.
\end{itemize}
We measured the performance of the allocation strategies for \BayesRS by the average
number of simulations each strategy needs until $ LB^\delta\ge 1-\alpha$, for
\KN and \PLUCK this goal had to be adapted. We also checked the empirical PCS,
i.e.\ the relative frequency of correct selections. For the experiments, we
implemented the different procedures with the free statistics software \texttt{R}.  In
\cite{Gorder2012} further studies show the efficiency of our allocation
strategies in the context of heuristic optimization methods like ant algorithms.

\subsection{Test setup}\label{subsec:TestSetup}
As our methods require normally distributed observations, we generated them by a
$\mathcal{N}_{L}(\MU,\SIGMA)$-random generator for different values of $ \MU$
and $ \SIGMA$, which, of course, were not known to the R\&S-strategies.  

Some parameters were fixed: the number of alternatives $ L=20$ (other values
showed similar behavior), the error probability $ \alpha=0.05$ and the
indifference zone parameter $ \delta=0.05$. The initial sample size was
$n_0=L=20$, the budget to be allocated in each iteration was $b= 10\cdot
L=200$ from which at least one simulation was allocated to each alternative in
each iteration. The parameter $ \nu_0$ for the prior probability of the
covariance matrix as in \eqref{eq:uninformed} showed best results in our setup
for $ \nu_0=L-1=19$ and was therefore fixed to this value throughout our tests.

  \begin{table}[htb]
    \centering
    \begin{tabular}[t]{l|l}
R\&S-case & $\MU$ \\\hline 
Best$_{1}$ &  $\mu_{1}=0,\;\;\mu_{2}=\cdots=\mu_{L}=1$\\
Best$_{10}$ &  $\mu_{1}=\cdots=\mu_{10}=0,\;\;\mu_{11}=\cdots=\mu_{L}=1$\\
Rank$_{10}$ & $\mu_{1}=0,\,\mu_{2}=1 ,\ldots, \mu_{10}=9, \mu_{11}=\cdots=\mu_{L}=10$\\
    \end{tabular}
    \caption{Different values for $\mu$  were used in the $ \MU$ -case ``ufc''.}
    \label{tab:MuCase}
  \end{table}

A \emph{basic scenario} consists of the following four variables  that are varied
in the tests:
\begin{enumerate}
\item We have examined three different \textbf{R\&S-cases} : in ``Best$_1$'' we
  want to select the best alternative, i.e.\ we use a target set $ A=\{1\}$ (see
  subsection~\ref{subsec:TargetAndSelection}), in ``Best$_{10}$'' we want to
  select the better half of the alternatives ($ A=\{1 ,\ldots, 10\}$) and in
  ``Rank$_{10}$'' we also want to rank these ten best alternatives.
\item Also three different \textbf{$ \MU$-cases} were used: In the unfavorable
  case ``ufc'', $ \MU$ is adapted to the R\&S-case chosen as described in Table
  \ref{tab:MuCase}. In the case ``inc'', $ \MU $ is the increasing sequence
  $\mu_{1}=0, \mu_{2}=1 ,\ldots, \mu_{L}=L-1$. In the case ``unif'', $ \mu_1
  ,\ldots, \mu_L$ are drawn randomly from the interval $ [0,100]$ with a minimal
  distance of at least $ \delta$. This
  is repeated $M_\mu:=15$ times, so that 15 simulations with different random $
  \MU$ are performed in the case ``unif''. 
\item To see how well our method works for different correlations among the
  observations, we created covariance matrices $
  \SIGMA=\big(\sigma_{ij}\big)_{i,j=1 ,\ldots, L}$ with a given joint
  correlation $ cor\in \{0.0,0.2,0.5, 0.7,0.9\}$.  To do so, variances $
  \sigma_{11} ,\ldots, \sigma_{LL}$ were chosen uniformly distributed in the
  interval $ [1,10]$, then we put $ \sigma_{ij}:=
  cor\sqrt{\sigma_{ii}\sigma_{jj}}$ for $ 1\le i,j\le L, i\not= j$. To include negative correlations, we
  complemented the above construction for $ cor \in \{-0.9, -0.5, -0.2\}$ with $
  \sigma_{ij}:= (-1)^{i-j} |cor|\sqrt{\sigma_{ii}\sigma_{jj}}$ resulting in a
  covariance matrix with alternating positive and negative entries.
  So we have eight different $ \SIGMA$\textbf{-cases}. For each case,
  $M_{cov}:=15$ different covariance matrices were constructed.
\item Finally, we distinguished two \textbf{distribution-cases}: in the case
  ``full'', we used the full joint posterior distribution 
  as  described above. In the case ``marginal'', we simply
  used the marginal posterior distributions of the alternatives and estimated
  the posterior means and variances by standard ML-estimators neglecting
  possible dependencies. This would be correct in the uncorrelated case with $
  cor=0.0$, but it is a great simplification in the other cases.
\end{enumerate}

All combinations of these parameters result in $ 3\times 3\times 8\times 2 =144$
different scenarios. For each scenario, $M_{cov}:=15$ covariance matrices $
\SIGMA$ were generated according to the $ \SIGMA$-case chosen and, if the $
\mu$-case ``unif'' was used, also $ M_\mu=15$ random vectors $ \MU$ were
generated. For each pair of $ (\MU,\SIGMA)$ thus chosen, $ M=20$ repetitions of
the different allocation strategies were performed. The averages from all these
trials are given below as \emph{mean no. of simulations}. We kept track of the random
seeds so that all strategies used the same random observations.

\subsection{The unfavorable  ${\mu}$-case}\label{ssec:ufc}

\begin{figure}[tb]  \centering
  \includegraphics[width=11cm,trim=0 40 20 55]{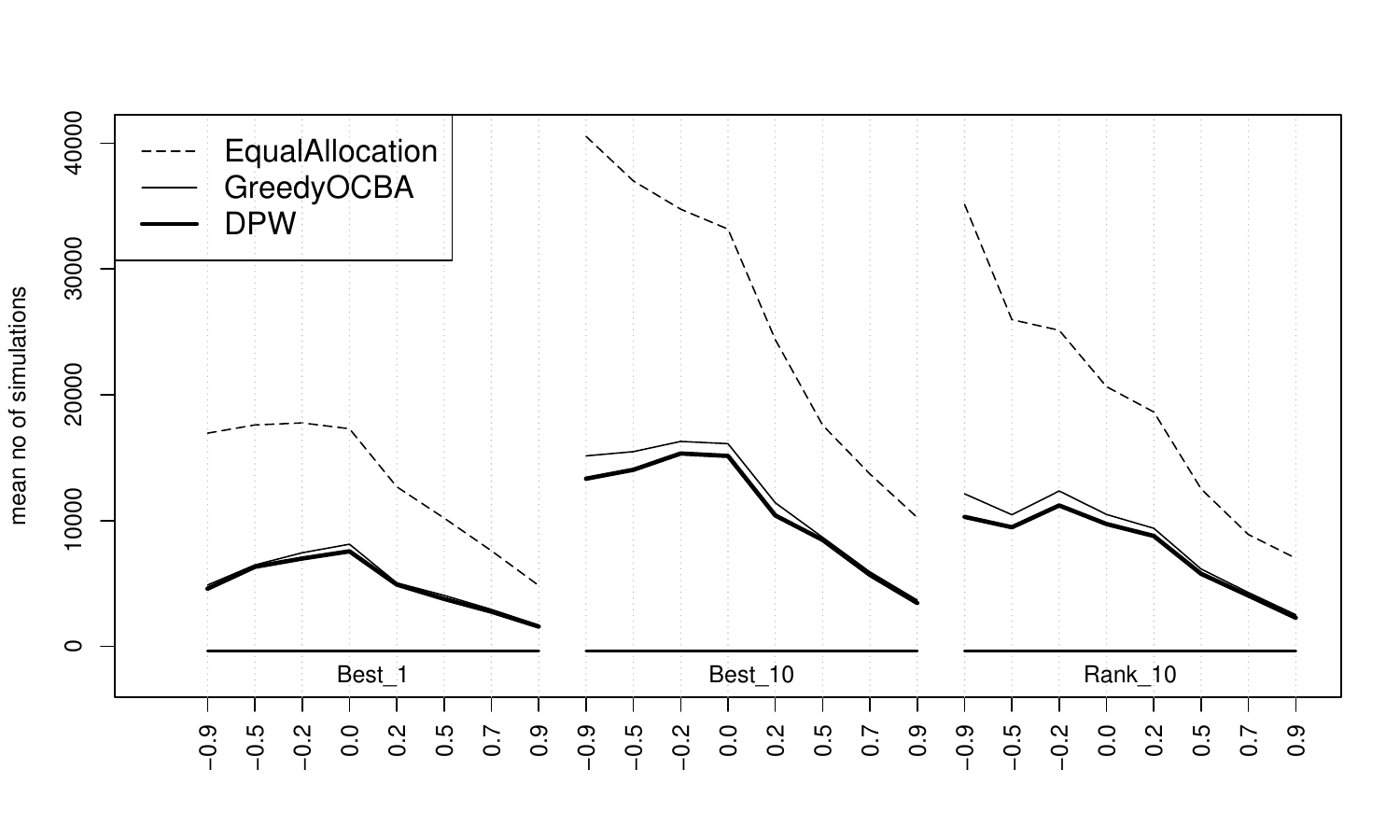}
  \caption{The unfavorable ${\mu}$-case: comparison of the allocation
    strategies \EqAlloc, \GreedyOCBA and \DPW for all R\&S-cases and all $
    \SIGMA$-cases.}
  \label{fig:ufc_CompFullDistrAbs}
\end{figure}
We first compare the three allocation strategies \EqAlloc, \GreedyOCBA and \DPW
within our algorithm \BayesRS. Here, \EqAlloc is mainly used to show how
difficult the problem instance is. In Figure \ref{fig:ufc_CompFullDistrAbs} the
mean number of simulations necessary to obtain a PCS $\ge 1-\alpha$ with
\BayesRS and the full posterior distribution is shown.  The $ x$-axis is divided
into the three R\&S-cases ``Best$_1$'', ``Best$_{10}$'' and ``Rank$_{10}$'', for
each of them the results for the eight $ \SIGMA$-cases $ cor= -0.9 ,\ldots, 0.9$
are given.
The number of simulations on the $ y$-axis is the mean over the $ M_{cov}=15$
different covariance matrices, each with the given correlation, and the $ M=20$
repetitions for each covariance matrix.

\begin{figure}[htb]
  \centering
  \includegraphics[width=12cm,trim=0 40 20 55]{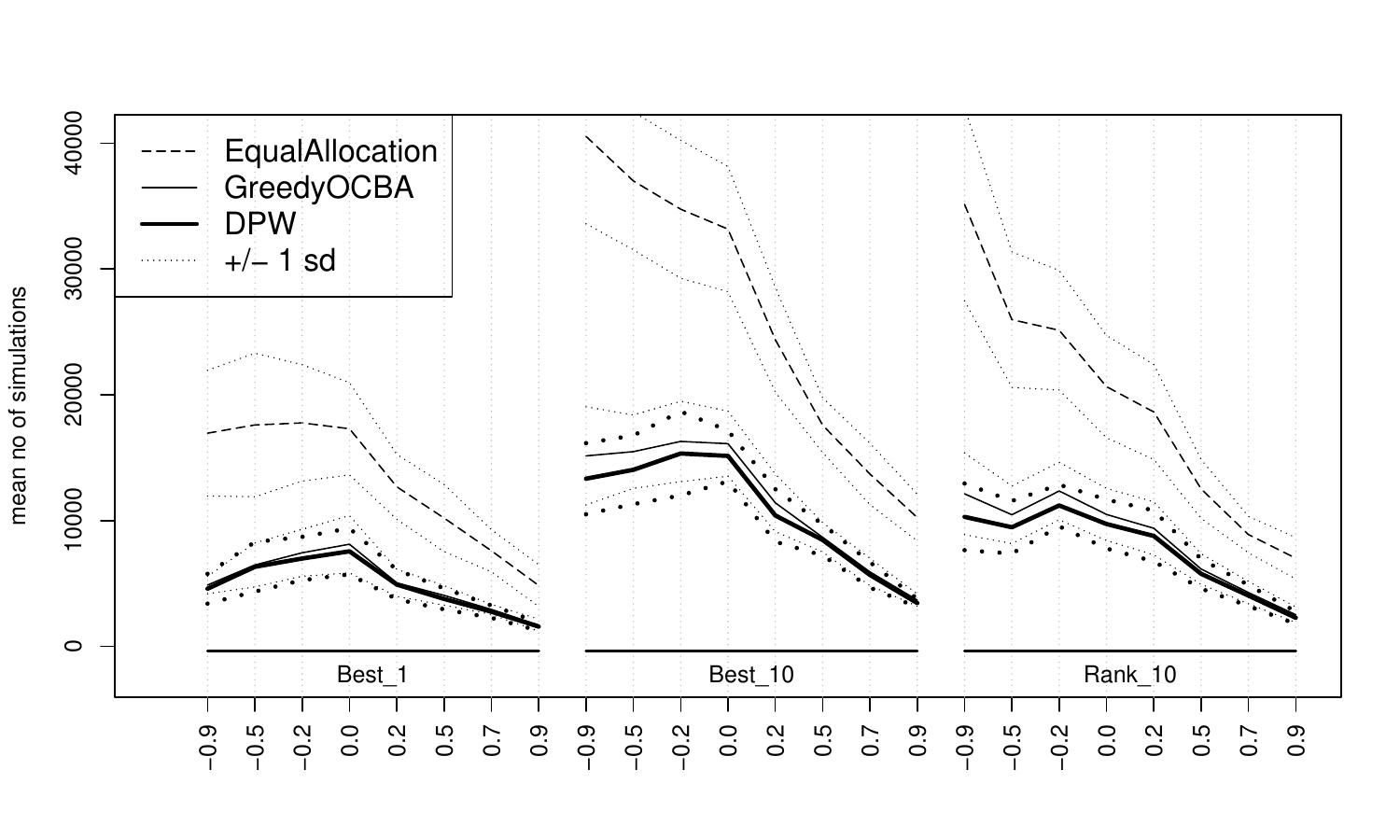}
  \caption{The unfavorable ${\mu}$-case:  The dotted lines indicate the standard deviation of the
    results.}
  \label{fig:ufc_CompFullDistrAbsSd}
\end{figure}
\begin{figure}[b]
  \centering
  \includegraphics[width=10cm,trim=0 40 20 55]{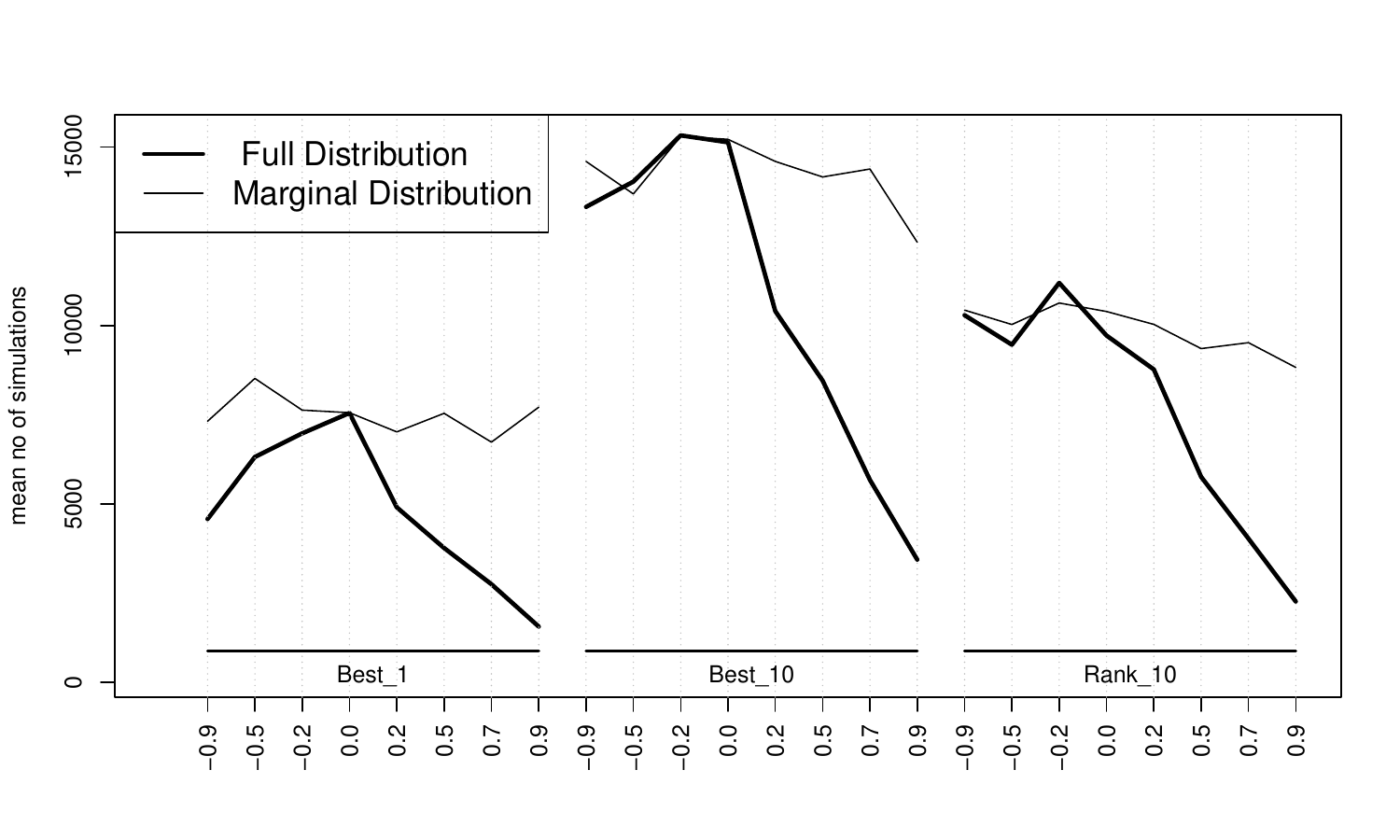}
  \caption{The unfavorable  ${\mu}$-case: comparison of the results based on
    full distributions with those using marginal distributions only (\BayesRS
    with allocation \DPW).}
  \label{fig:ufc_CompMargFullAbs}
\end{figure}

Figure \ref{fig:ufc_CompFullDistrAbs} clearly shows that the intelligent
allocation rules are much more efficient than the simple equal allocation. Also
in all scenarios, our new allocation rule \DPW performed slightly better than
the more classical \GreedyOCBA on the average.  However, the advantage of \DPW
over \GreedyOCBA is not significant as can be seen from Figure
\ref{fig:ufc_CompFullDistrAbsSd} where curves with $ \pm$ one standard deviation
have been added.

Figure \ref{fig:ufc_CompMargFullAbs} compares the two distribution cases for
\BayesRS with \DPW. The bold curve gives the mean number of observations
based on dominance probabilities as given in \eqref{eq:DP} (i.e.\ case
``full''), the light curve are results that are obtained from dominance probabilities based on
the marginal distributions of the alternatives only (case ``marginal''), neglecting
the possible dependencies.  Figure \ref{fig:ufc_CompMargFullAbs} clearly shows
that it is indeed worthwhile to calculate the full posterior distributions in
order to save simulations. In particular if observations are positively
correlated, the variance of the difference of two alternatives as used in
\eqref{eq:DP} is reduced under the full distribution leading to much smaller
sample sizes.

The question remains if the reduced sample sizes in the case ``full'' are large enough to select the
right alternatives. We skip the picture of the empirical PCS for this case as it
is equal to one for all strategies  (except for one case of \EqAlloc). This also
shows that our assumptions are quite conservative.

\subsection{The increasing $ \mu$-case}
Next we look at the $ \MU$-case ''inc'', the results are quite similar to the $
\MU$-case ``ufc''. Figure \ref{fig:inc_CompFullDistrAbs} shows the mean number
of simulations in the ``full'' distribution case and again, \DPW performs
slightly better than \GreedyOCBA. The comparison of the two distribution cases
``full'' and marginal is similar as in the case ``ufc''. Figure
\ref{fig:inc_empPCS} shows that the empirical PCS is well above the required
95\%.
\begin{figure}[tb]
  \centering
  \includegraphics[width=12cm,trim=0 40 20 55]{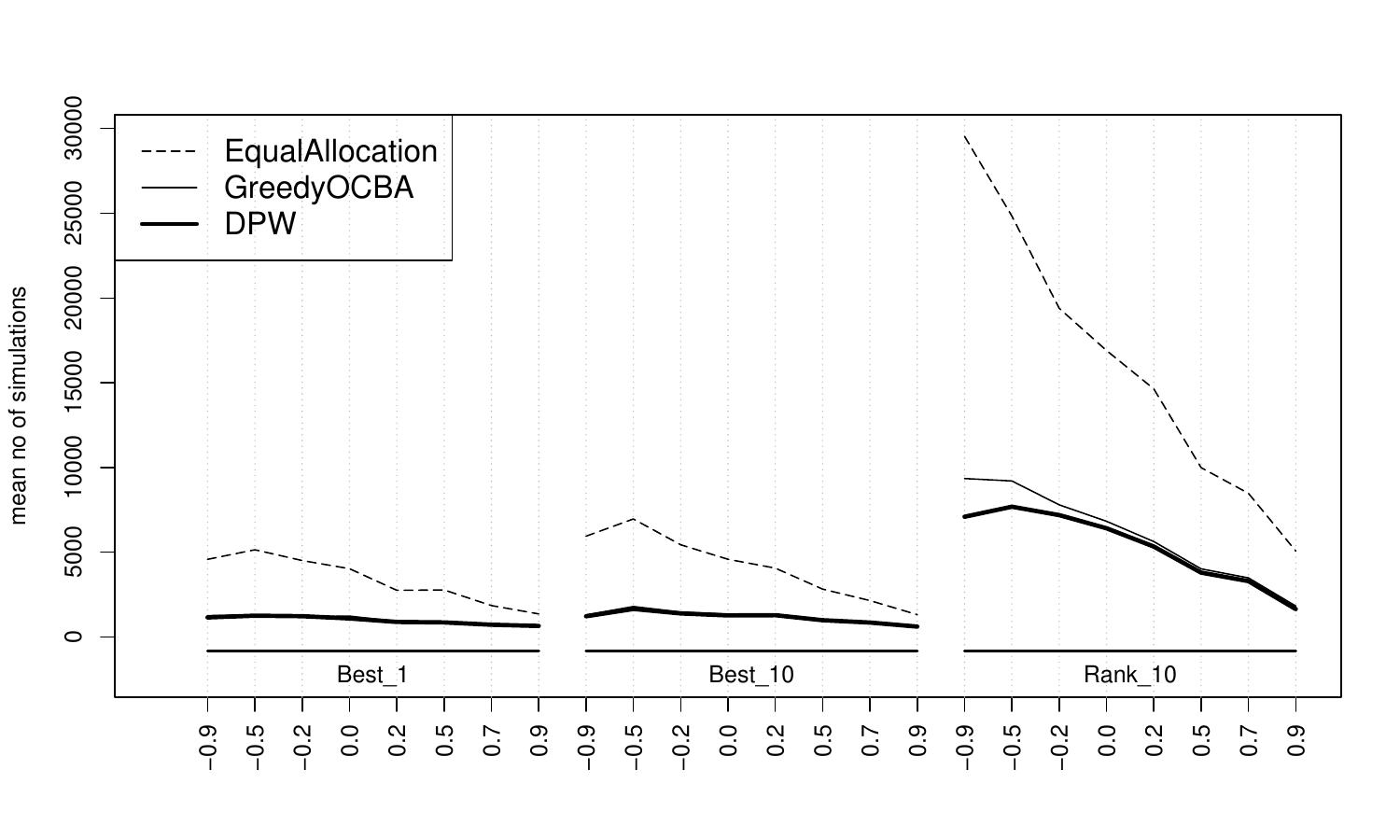}
  \caption{$ \MU$-case ``inc'': Comparison of the allocation strategies
    \EqAlloc, \GreedyOCBA and \DPW using the full
    posterior distribution of the unknown means.}
  \label{fig:inc_CompFullDistrAbs}
\end{figure}

\begin{figure}[hb]
  \centering
  \includegraphics[width=11cm,trim=0 40 20 55]{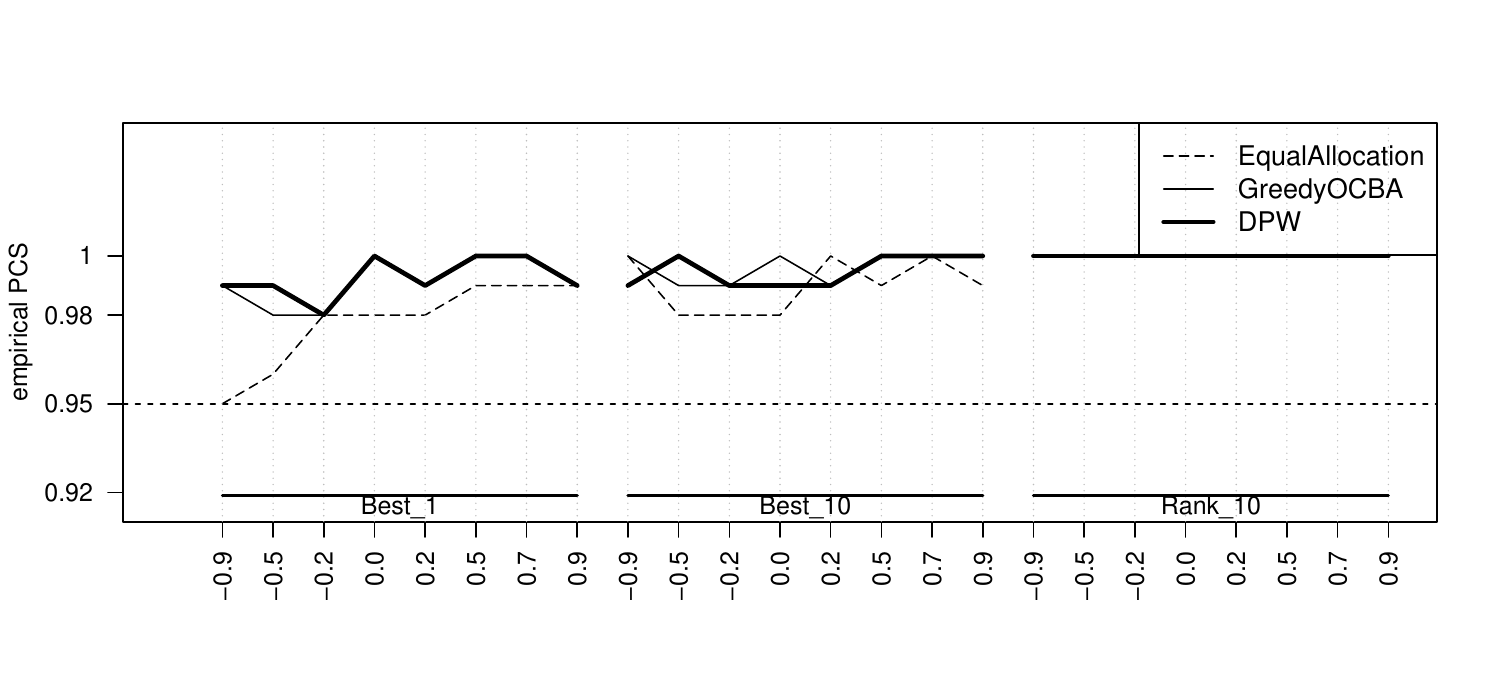}
  \caption{$ \MU$-case ``inc'': the empirical PCS of the three different
    allocation strategies for ``full'' distribution-case.}
  \label{fig:inc_empPCS}
\end{figure}

\subsection{The uniform  $ \mu$-case}
\begin{figure}[htb]
  \centering
  \includegraphics[width=12cm,trim=0 40 20 55]{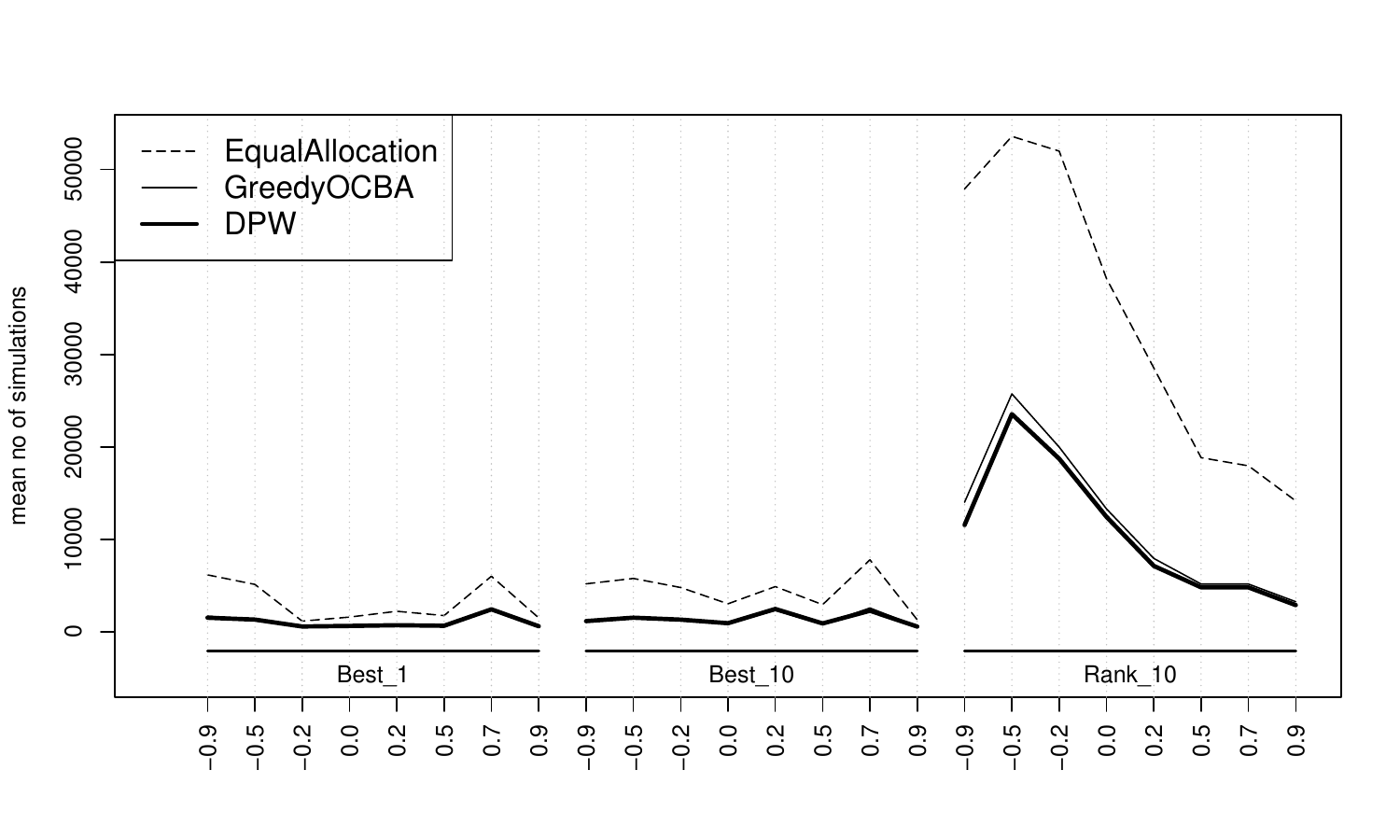}
  \caption{$ \MU$-case ``unif'': Comparison of the allocation strategies
    \EqAlloc, \GreedyOCBA and \DPW using the full
    posterior distribution of the unknown means.}
  \label{fig:unif_CompFullDistrAbs}
\end{figure}

In the ``unif''-case, results are means over $ M_{\mu}=15$  different $ \MU$  drawn
randomly from $ [0,100]^L$, over $ M_{cov}=15$ random covariance matrices and $
M=20$ repetitions. To make things comparable, we adapted the indifference zone
parameter to $0.05\times\min\{\mu_i-\mu_j\mid (i,j)\in \rho_{AB}\}$, i.e.\ we
adapted it to the actual minimal difference appearing in the random mean $\MU$.

Results vary considerably depending on the difficulty of $ \MU$, so that the
means give only a rough picture of the performance.  Figure
\ref{fig:unif_CompFullDistrAbs} shows that also for this case our new strategy
\DPW is at least as good as \GreedyOCBA though it uses less calculations. For
the most difficult task 'Rank\_10', the simple \EqAlloc could not solve all
instances within our limit of 120\,000 simulations, therefore it missed the
empirical PCS of 95\% in some cases.

\subsection{Comparison to the procedure \KN}
\begin{figure}[tb]
  \centering
   \includegraphics[width=11cm,trim=0 25 20 35]{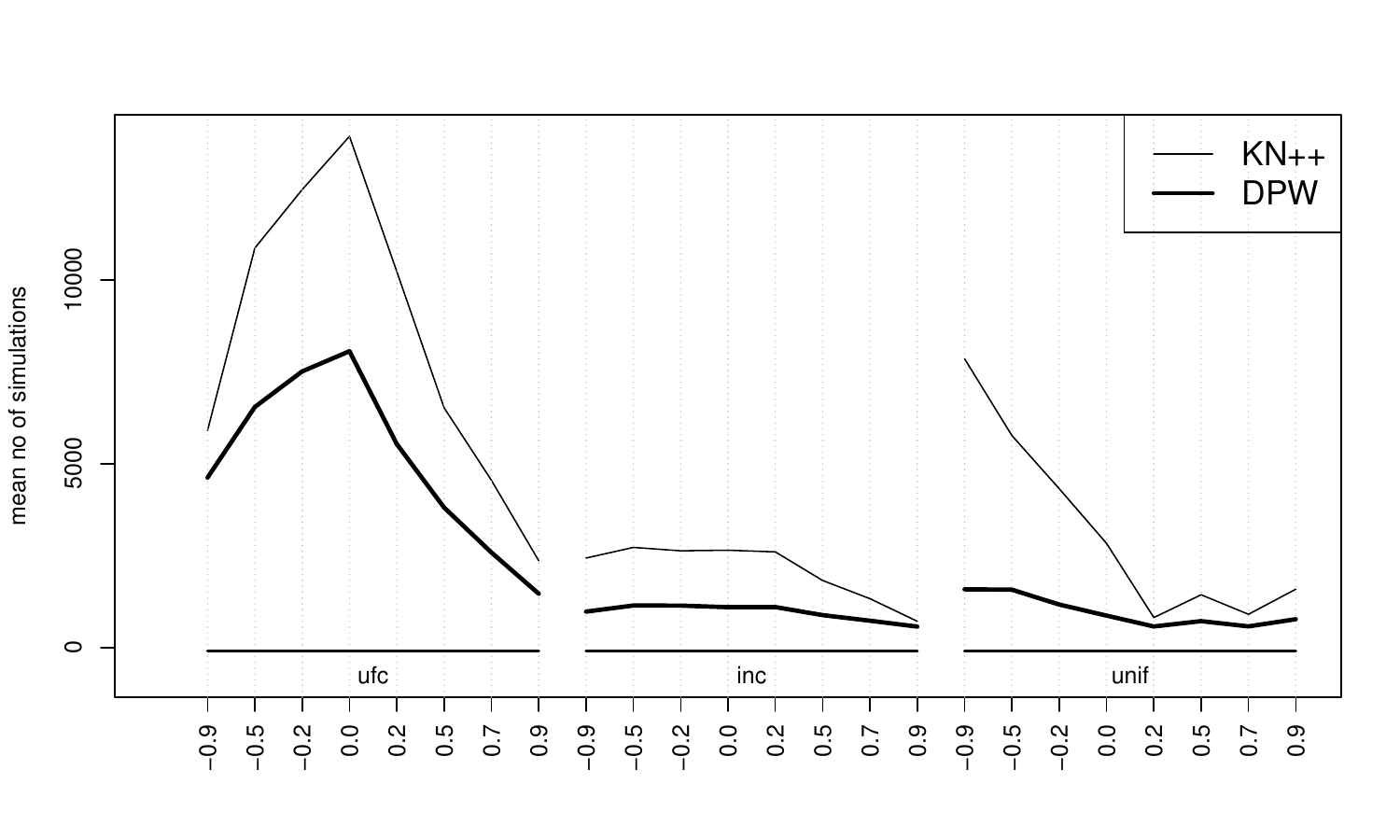}
   \caption{Comparison of \BayesRS with \DPW and $ \cal KN++$ for the R\&S-case
     Best$_1$. Here, we used a relaxed indifference parameter $
     \delta_{\text{\KN}}=0.5$ for \KN
     instead of $ \delta=0.05$ as in \BayesRS.    }
  \label{fig:CompDPWKN}
\end{figure}

In \cite{kim2006asymptotic} (see also \cite{KimNelsonUeber}) the sequential
procedure \KN is introduced. It uses a set of active alternatives each of which
is simulated once in each iteration. Then, alternatives that are inferior
to one of the other active ones are excluded from the active set and from
further simulations. The procedure stops as soon as there is only one active alternative
left  which is then selected as 'best'. The observations for the active
alternatives may be correlated but there is no need for re-use of random seeds
as the exclusion of alternatives forces a monotone pattern of missing values.

We restricted ourselves to the RS-case Best$_1$ (target set $ A=\{1\}$). \KN was
allowed to run until only one alternative was left, \BayesRS with allocation
\DPW used the stopping criterion PCS $\ge 1-\alpha$ as before. In many cases \KN
could not finish within our limit of 120\;000 trials, therefore we biased the
set-up in favor of \KN: instead of the indifference parameter $ \delta=0.05$ it
uses $ \delta_{\text{\KN}}=0.5$.  Figure \ref{fig:CompDPWKN} shows the
results. As we have only one RS-case, we collected all three $ \MU$-cases on the
$ x$-axis, each with the eight $ \SIGMA$-cases. It turned out that even with an
indifference parameter ten times larger, \KN needed much more simulations to
select the best alternative than \BayesRS, in particular in the unfavorable $
\MU$-case ``ufc''. The empirical PCS war almost one for both procedures.

\subsection{Comparison to the  procedure \PLUCK}
\PLUCK (\textbf{p}rojected \textbf{l}earning of \textbf{u}nknown
\textbf{c}orrelation with \textbf{k}nowledge gradients) as defined in
\cite{qu2012ranking} is a fully sequential procedure that allocates just one
simulation in each iteration. It aims to select the alternative with the largest
mean value, i.e.\ we have to use target set $ A=\{L\}$ in our \BayesRS procedure.

\PLUCK is a Bayesian procedure that assumes a Normal-InverseWishart joint prior
distribution for $ W$ and $ S$, i.e.\ the conditional distribution of the means
$ W$, given that $ S=\SIGMA$, is a Normal distribution $ \NN_L(\theta,\SIGMA/q)$
where $ \theta\in \R^L$ and $ q \in \R$ are known parameters. $ S$ has an
inverse Wishart distribution with known parameter $ \gamma\in \R$ and known
scale matrix $ \Gamma$. The posterior distribution after the observation of a
single alternative is approximated in \cite{qu2012ranking} by a
Normal-InverseWishart distribution which is used as prior distribution for the
next iteration. The parameters of this approximation are
updates of the prior parameters depending on the single observation in a rather
complicated way. In particular, for the posterior update $ \gamma^*$ of the
scalar $ \gamma$ a numerical solution to an equation is needed. We replaced this
by a rough approximation also mentioned in \cite{qu2012ranking} and put $
\gamma^*=\gamma+1/L$.

Instead of maximizing the PCS, \PLUCK uses the  value-of-information
approach. In each iteration $ n$, it determines the mean $
\hat{\theta}_n=(\hat{\theta}_{n,1} ,\ldots, \hat{\theta}_{n,L})$ of the
(approximate) posterior distribution of $ W$. Then $ \Theta_n:=\max_{j\in \L}
\hat{\theta}_{n,j}$ is the present estimate of the best (in this case largest)
unknown mean and the maximizing alternative $ j$ is the alternative to be
selected in the present iteration.  The \emph{value of information} for an
alternative $ i$ is the difference between $ \Theta_n$ and the expectation of $
\Theta_n$ after an additional simulation has been performed with alternative
$i$.  The next actual simulation is then allocated to an alternative  that
has the largest value of information. The performance of \PLUCK in an experiment
is evaluated by the \emph{opportunity cost} which is the difference between the
true largest mean and the mean selected by \PLUCK.

\begin{figure}[tb]
  \centering
  \includegraphics[width=12cm,trim=20 40 20 55]{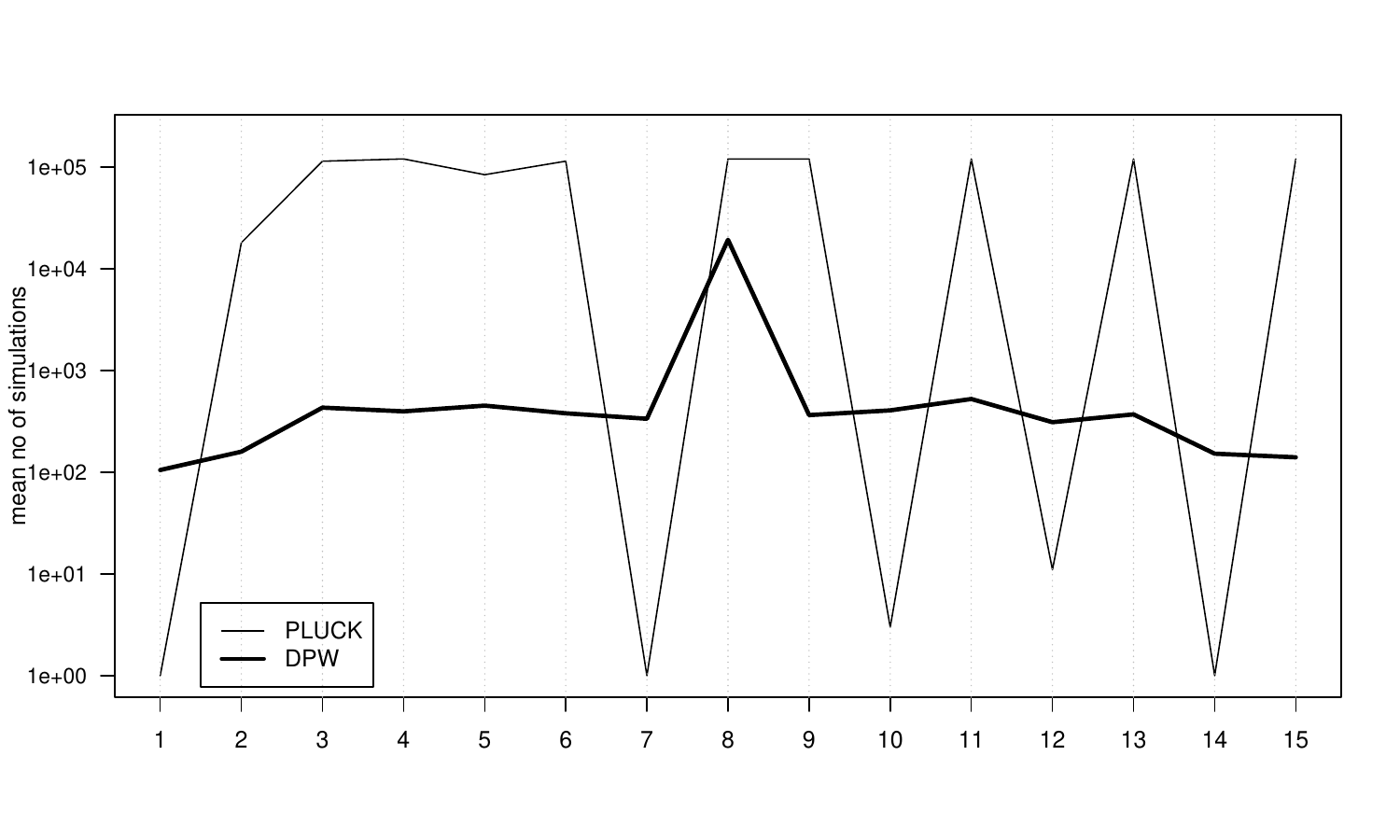}
  \caption{The results for the increasing prior mean $ \theta=(0,0.2,0.4 ,\ldots, (L-1)\cdot
    0.2)$. Shown are the mean number of simulations for the $ 15$
    random covariance matrices $ \SIGMA$ drawn from an inverse Wishart
    distribution.  Note that the $ y$-axis uses a logarithmic scale.  }
  \label{fig:PLUCKinc}
\end{figure}
\begin{figure}[htb]
  \centering
 \includegraphics[width=12cm,trim=0 40 40 55]{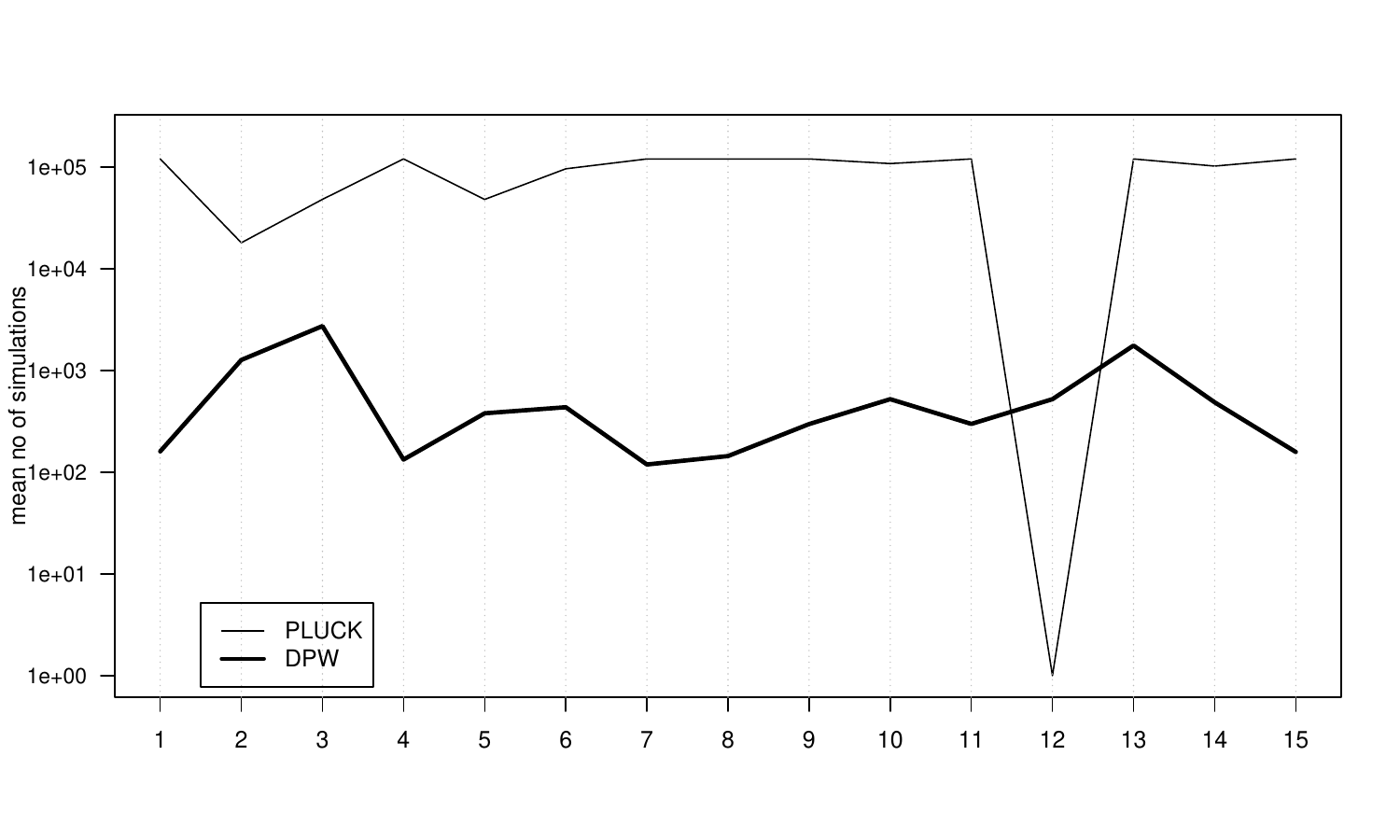}
 \caption{The results for the unfavorable prior mean $ \theta=(0 ,\ldots,
   0,0.2)$. Shown are the logarithms of the mean number of simulations for the $ 15$ random
   covariance matrices $ \SIGMA$.  }
  \label{fig:PLUCKlfc}
\end{figure}

\PLUCK needs the parameters $ \theta,\Gamma,\gamma$ and $ q$.  To make \PLUCK
comparable to our set-up, we chose as prior parameter $ \theta$ the vector
determined by our $ \MU$-case, i.e., as unfavorable, increasing or random, see
Subsection \ref{subsec:TestSetup}.  $ \Gamma$ was drawn randomly from a Wishart
distribution (with random scale matrix). We set $ \gamma=L+1$ and $ q=1$ and then
simulated the covariance matrix $ \SIGMA$ with an inverse Wishart distribution
with parameters $ \gamma$ and $ \Gamma$.  The mean $\MU=(\mu_1 ,\ldots,
\mu_L)$ is drawn from a Normal distribution $ \NN_L(\theta,\SIGMA)$ (as $ q=1$),
so that the pair $ (\MU,\SIGMA)$ is drawn from a Normal-InverseWishart
distribution as required. In this way, $ M_{cov}=15$ pairs $( \MU,\SIGMA)$ are
produced for each of the $ \MU$-cases ``ufc'' and ``inc''. For the case
``unif'', we repeated these steps for the $ M_{\mu}=15$ randomly drawn values of
$ \theta$. Then, simulations were repeated $ M=20$ times as before.  

Given a pair $ (\MU,\SIGMA)$, the observations $ \X$ are drawn from a Normal
distribution $ \NN_L(\MU,\SIGMA)$ as before. For \PLUCK, the single samples for
each iteration are independent of each other, whereas for \BayesRS we used the
sampling scheme as described in \ref{subsec:IterativeAllocation}.

We continued the iterations in \PLUCK until the opportunity cost was smaller
then the indifference parameter $ \delta=0.05$.  For our procedure \BayesRS with
allocation strategy \DPW we assumed an uninformative prior as before and ran it
until PCS $ \ge 1-\alpha$. This means that \PLUCK was allowed to know when it
hit the true value whereas \BayesRS had to ensure  PCS $ \ge 1-\alpha$.

It turned out, that these test cases were too simple, \PLUCK selected the true
value after one or two steps in most cases, though in some cases it could not
find a solution within our limits of $120\,000$ observations. \BayesRS had to
make the minimal number of $ n_0=20$ complete observations before it could
start. Still on the average, \BayesRS was much better due to the few outliers of
PLUCK with $ 120\,000$ observations.

To obtain a more reliable comparison, we restricted ourselves to the prior distributions
that were formed after the $ \MU$-cases ''inc'' and ''ufc''. To make these more
difficult, we changed the priors from $ (0,1,2 ,\ldots, L-1)$ and $ (0
,\ldots, 0,1)$ to $ (0,0.2,0.4 ,\ldots, (L-1)\cdot 0.2)$   and $ (0 ,\ldots,
0,0.2)$, i.e.\ we decreased the step-size from $ 1$ to $ 0.2$ to make the true
means more difficult to distinguish. At the same time,
we extended \BayesRS by computing the exact posterior distributions during the
first $ n_0$ complete observations (using \cite{DeGroot2004d},10.3) and stopped
when   PCS $ \ge 1-\alpha$. This is a useful extension for very
simple cases. 

Figures \ref{fig:PLUCKinc} and \ref{fig:PLUCKlfc} show the results for these two
cases. Here the averages over $ M=20$ repetitions for each of the randomly drawn
$ M_{cov}=15$ covariance matrices $ \SIGMA$ is given. As can be seen, in some
cases \PLUCK found the true solution in a single step, whereas in most others,
it could not find it within the limit of $ 120\,000$ observations.
Consequently, the empirical PCS was $ 0$ for these cases, whereas our \BayesRS
had PCS $ \ge 0.95\%$ and even PCS $ =1$ in almost all cases.  Probably, \PLUCK
is more suited for a larger number of alternatives, where the opportunity cost
has more meaning than in our case.

% Figure \ref{fig:PLUCK} shows on the lefthand side the mean number of simulations
% needed to meet the stopping criteria in the three $ \MU$-cases.  \PLUCK is
% outperformed even by the simple \EqAlloc strategy. Allocating just a single
% trial in each iteration seems to be too greedy, \PLUCK tends to concentrate on
% alternatives that seem promising from the present limited knowledge without
% exploring other alternatives. Here, \PLUCK could not finish within our limit of
% 120\,000 trials in 19\%, 7\% and 13\% of the instances for ``ufc'', ``inc'' and
% ``unif'', respectively.  The righthand side of Figure \ref{fig:PLUCK} shows that
% even with much more trials, \PLUCK hits the correct selection with a much
% smaller frequency then the other two approaches.

\section{Conclusion and future work}\label{sec:Conclusion}

In this paper we presented a new sequential Bayesian R\&S procedure with support
for common random numbers. Based on an approximation of the posterior
distribution of the unknown mean and covariance, the simulation effort could be
allocated to alternatives for which insufficient data were available for a
pairwise comparison. 

Extensive experiments showed the practicability of this
approach. In particular, it proved superior to the strategy \PLUCK that also
tries to evaluate the posterior distributions but uses only one trial in each
iteration and might get stuck in case the (known) prior parameters are misleading.

In our future work we will extend this concept to the selection of multivariate
parameters as they occur in multi-criteria optimization problems. Essential parts
of this problem were solved in \citet{Gorder2012}.

%%%%%%%%%%%%%%%%%%%%%%%%%%%%%%%%%%%%%%%%%%%%%%%%%%
\appendix
\section{Derivation of the Posterior Distribution}

From \cite{schafer1997analysis}, 5.2.4, (or many other texts on Normal
distributions) we see, that if $ (X_{1,k} ,\ldots, X_{L,k})$ is $
\NN_L(\MU,\SIGMA)$-distributed, then the conditional distribution of $ X_{ik}$
given $ (X_{1,k} ,\ldots, X_{i-1,k})=(x_{1,k} ,\ldots,x_{i-1,k})=:\x_{[<i,k]}$
is a one-dimensional Normal distribution with mean
\begin{align}
 \tilde{\mu}_{(i-1)}(\x_{[<i,k]}) &:= \mu_i
 +\beta_i(\x_{[<i,k]}-\MU_{[<i]})\quad \text{  and variance } \label{eq:DefMuTilde}\\
 \tilde{\sigma}_{(i-1)}&:= \sigma_{ii} - \beta_i \SIGMA_{[<i]} \beta_i^T\notag
\end{align}
where $ \beta_i$ was defined in \eqref{eq:DefBeta}. This allows to rewrite the
density of $\NN_L(\MU,\SIGMA) $ as a product of one-dimensional Normal densities
with parameters $\big(\tilde{\mu}_{(i-1)}(\cdot)$, $\tilde{\sigma}_{(i-1)}\big),
i=1 ,\ldots, l$, where for $ i=1$ we put $
\tilde{\mu}_{(0)}=\MU_1,\tilde{\sigma}_{(0)}=\sigma_{11}$.

Using this re-parameterization, the likelihood function $l(\MU,\x) $ of $ \MU$ for
a possibly incomplete observation $ \X=\x$ with $ n_1\ge n_2\ge \cdots \ge n_L$
as in Theorem \ref{the:posterior} is a product of one-dimensional Normal
densities which can be rearranged to (see e.g.\ \cite{schafer1997analysis}, 6.5)
\begin{align}
&l(\MU,\x) \propto \\
%\phi_1(\bar{\x}_{1};\ \mu_1,\sigma_{11}/n_1) \cdot \prod_{i=2}^L \phi_1(\bar{\x}_i;\ \mu_i+\beta_i(\bar{\x}_{[<i]}^{(n_i)} - \MU_{[<i]}), \tilde{\sigma}_{(i-1)}/n_i) \notag\\
%&= 
&\phi_1(\mu_1;\ {\xx}_{1} ,\sigma_{11}/n_1) \cdot \prod_{i=2}^L \phi_1(\mu_i;\ {\xx}_i +\beta_i(  \MU_{[<i]}-{\xx}_{[<i]}^{(n_i)}), \tilde{\sigma}_{(i-1)}/n_i)\notag.
\end{align}
If we assume $ \SIGMA$ to be known and use the uninformative prior $
\pi(\MU)\equiv 1$ for $ W$, then this likelihood is also the posterior density of $
W$ given $ \X=\x$. It is a density of a $ L$-dimensional Normal distribution $
\NN_L(\NU,\LAMBDA)$ and its factors are the conditional densities of $ W_i$
given $ W_1 ,\ldots, W_{i-1},\X$. Using standard properties of the conditional
expectation and conditional covariances, the mean $ \NU$ is obtained as
\begin{align*}
   \nu_1&= \E [W_1 \mid \X=\x] = {\xx}_{1} \label{eq:Nu1}\\
 \nu_i&= \E [W_i \mid \X=\x]\ = {\xx}_i +\beta_i( \NU_{[<i]} - {\xx}_{[<i]}^{(n_i)})\qquad \text{ for } i=2 ,\ldots, L.\end{align*}
which proves \eqref{eq:DefNuTheorem}. Similarly, we obtain $ \LAMBDA=(\Lambda_{ij})_{1\le i,j\le L}$ from
\begin{align*}
\Lambda_{11}&=\V[W_1\mid \X=\x]= \frac{\sigma_{11}}{n_1}\\
\Lambda_{ii}&=\V[W_i\mid \X=\x]\\
&= \E\Big[ \V[ W_i|\X, W_{[<i]}] \;\Big|\, \X=\x \Big]+\V \Big[\E[W_i\mid \X, W_{[<i]}]\;\Big|\, \X=\x\Big]\\
&= \frac{1}{n_i}(\sigma_{ii}-\beta_i\SIGMA_{[<i]}\beta_i^T) + \beta_i \
\LAMBDA_{[<i]}\ \beta_i^T,\quad \text{  and for $ 1\le k<i$  } \\
 \Lambda_{ki}&=\cov\Big[W_k, W_i \ \big|\  \X=\x\Big] \\
  &= \beta_i \Big(\cov [W_k, W_l\ \big|\ \X_{[
<i,\Pkt]}=\x_{[<i,\Pkt]}]\Big)^T_{l=1 ,\ldots, i-1} =   \LAMBDA_{[k,<i]}\beta_i^T.\notag
\end{align*}
which proves \eqref{eq:DefLambdaTheorem}.

%\bibliographystyle{alpha}
%%\bibliographystyle{plainnat}
%\bibliographystyle{acmtrans-ims}
%\bibliographystyle{imsart-nameyear}
%\bibliographystyle{elsarticle-harv}
%%\bibliography{paper,auto_k}

\end{document}